\documentclass[reqno]{amsart}
\usepackage{graphicx,epsfig}
\usepackage{enumitem}
\usepackage[numbers,sort&compress]{natbib}
\usepackage{hyperref}
\usepackage{soul}
\usepackage[foot]{amsaddr}
\usepackage[final,color,notref,notcite]{showkeys}
\usepackage{tikz-cd}
\usepackage{booktabs}

\usepackage[charter]{mathdesign}
\usepackage[scaled=.96,lf]{XCharter}

\theoremstyle{plain}
\newtheorem{thm}{Theorem}
\newtheorem{cor}[thm]{Corollary}
\newtheorem{lem}[thm]{Lemma}
\newtheorem{prop}[thm]{Proposition}

\theoremstyle{definition}
\newtheorem{defi}{Definition}

\theoremstyle{remark}
\newtheorem*{rem}{\bf Remark}

\definecolor{labelkey}{rgb}{0,.56,.7}

\DeclareMathAlphabet{\pazocal}{OMS}{zplm}{m}{n}   

\def\bbZ{\mathbb{Z}}
\def\bbR{\mathbb{R}}

\DeclareSymbolFont{Eulerscripteusm10}{U}{eus}{m}{n}
\SetSymbolFont{Eulerscripteusm10}{bold}{U}{eus}{b}{n}
\DeclareMathSymbol{\euI}{\mathord}{Eulerscripteusm10}{"4A}

\newcommand*{\at}{@}
\newcommand{\papa}[2]{\frac{\partial#1}{\partial#2}}
\newcommand{\dif}[1]{\mathrm{\,d} #1}             
\newcommand{\diff}[1]{\mathrm{\,d}^2 #1}
\newcommand{\nn}{\nonumber}

\def\dg{\dagger}
\def\df{\overset{\mathrm{df}}{=}}
\newcommand{\ket}[1]{\mathop{|#1\rangle}\nolimits}

\newcommand{\brk}[2]{\langle #1 | #2 \rangle}
\newcommand{\kbr}[2]{| #1\rangle\!\langle #2 |}

\def\ran{\rangle}
\def\lan{\langle}

\DeclareMathOperator*{\range}{ran}

\newcommand{\Tr}[1]{\mathop{{\mathrm{Tr}}_{#1}}}
\newcommand{\id}{\mathop{{\mathrm{id}}}\nolimits}

\def\a{\alpha}
\def\b{\beta}
\def\g{\gamma}
\def\d{\delta}
\def\D{\Delta}

\def\ve{\varepsilon}
\def\vr{\varrho}
\def\vp{\varphi}
\def\ka{\kappa}

\def\s{\sigma}

\def\vt{\vartheta}

\allowdisplaybreaks

\sodef\so{}{.065em}{.4em plus1em}{2em plus.1em minus.1em}

\makeatletter
\def\@setemails{%
\ifnum\theg@author > 1
\mbox{\hspace{-4mm}{\itshape E-mail}:\space}{\ttfamily\emails}.
\else
\mbox{\hspace{-4mm}{\itshape E-mail}:\space}{\ttfamily\emails}.
\fi%
}
\makeatother

\makeatletter
\def\@setfoot@addresses{
\def\author##1{\hspace{-4mm}\setlength{\parindent}{0pt}}%
\def\\{\unskip, \ignorespaces}%
\newif\if@firstaddr
\@firstaddrtrue
\def\address##1##2{%
\if@firstaddr\@firstaddrfalse\else\par\fi
\@ifnotempty{##1}{(\ignorespaces##1\unskip) }%
{\scshape\ignorespaces##2}%
}%
\def\email##1##2{}%
\def\curraddr##1##2{}%
\def\urladdr##1##2{}%
\addresses
}
\makeatother

\makeatletter
\def\@setkeywords{%
  {\itshape\hspace{-4.8mm} \keywordsname.}\enspace \@keywords\@addpunct.}
\makeatother

\begin{document}

\title{A security proof of continuous-variable QKD using three coherent states}

\author{Kamil Br\'adler}
\email{kbradler\at uottawa.ca}
\address[Kamil Br\'adler]{Department of Mathematics and Statistics, University of Ottawa, Ottawa, Canada}

\author{Christian Weedbrook}
\address[Christian Weedbrook, Kamil Br\'adler]{CipherQ, 10 Dundas St E, Toronto, M5B 2G9, Canada}


\begin{abstract}
We introduce a new ternary quantum key distribution (QKD) protocol and asymptotic security proof based on three coherent states and homodyne detection. Previous work had considered the binary case of two coherent states and here we nontrivially extend this to three. Our motivation is to leverage the practical benefits of both discrete and continuous (Gaussian) encoding schemes creating a best-of-both-worlds approach; namely, the postprocessing of discrete encodings and the hardware benefits of continuous ones. We present a thorough and detailed security proof in the limit of infinite signal states which allows us to lower bound the secret key rate. We calculate this is in the context of collective eavesdropping attacks and reverse reconciliation postprocessing. Finally, we compare the ternary coherent state protocol to other well-known QKD schemes (and fundamental repeaterless limits) in terms of secret key rates and loss.
\end{abstract}

\maketitle

\thispagestyle{empty}

\section{Introduction}\label{sec:intro}

Quantum key distribution (QKD)~\cite{Scarani2009,Lo2014}, in principle, provides the most secure form of quantum safe cybersecurity, i.e., protection against a quantum computing attack. As opposed to post quantum cryptography~\cite{Chen2016}, which is based on computationally secure mathematics, QKD exploits the laws of quantum physics to achieve, at least in theory, unbreakable codes. Since QKD was first suggested in 1984, many advances have taken place; from theoretical to proof-of-principle experiments to field tests and even the forming of companies.

Even though this seems like the end of the story there are still many advances being made in all of these areas. To this point, in this paper, we look at creating a best-of-both worlds approach to QKD by combining the beneficial practical aspects of the two main implementations of QKD: those using discrete variables (DVs)~\cite{Scarani2009} and those using continuous variables (CVs)~\cite{Weedbrook2012,diamanti2015distributing}. To be more specific, we would like to use the simpler encoding and decoding methods from DV QKD but at the same time leverage the simpler and more affordable room temperature hardware components of CV QKD.

Recently, the ultimate (optimal) limit for a lossy bosonic channel was discovered and is given by the PLOB bound~\cite{pirandola2015fundamental}. An interpretation of this result is that no QKD protocol can go beyond this bound without a quantum repeater. In terms of key rate as a function of channel loss (cf. for instance with Fig. 6 of \cite{pirandola2015fundamental}) this corresponds to a CV QKD Gaussian protocol with reverse reconciliation using a quantum memory at Alice's side and heterodyne at Bob's side~\cite{Garcia2009}. In terms of implementations, below this optimal bound lies the single photon BB84 protocol~\cite{Lo2005}. Both of these two protocols are in terms of the ideal case, i.e., perfect sources and perfect detectors. However, when one considers the realistic version of these two (in the case of DV QKD this corresponds to the decoy state scheme~\cite{Lo2005a,Wang2005}), both become remarkably similar in terms of key rates as a function of loss; except for a slight advantage in key rates for CVs in the low-loss regime and a slight distance advantage in DVs for the high-loss regime. In this realistic scenario, both the DV and the CV QKD schemes sit below the PLOB bound. Ideally we would like to either: (1) find a (realistic) protocol above these two protocols or (2) have a protocol similar to these protocols in terms of key rates but one that leverages the practical benefits of both schemes.

With that in mind, we consider a protocol first introduced in 2009 by Zhao et al.~\cite{zhao2009asymptotic} that uses binary-phase shift-keying (BPSK) of coherent states, $\ket{\a}$ and $\ket{-\a}$, along with homodyne detection. Unfortunately, as one can see, the performance of this protocol is below that of the realistic BB84 with decoy states and the realistic Gaussian modulation CV scheme. In this paper, we consider a ternary-phase shift-keying (TPSK) of coherent states, $\ket{\a}_i$ where $i = 0,1,2$, with homodyne detection. Here each of the three coherent states are phase shifted in phase space by $120^{\circ}$, cf. Fig.~\ref{fig:Phase_space}. One may ask the question, what is the motivation of going from two coherent states to three coherent states? Or perhaps why not go to more coherent states straightaway? In terms of the second question, this is easily answered by considering the Zhao et al. paper~\cite{zhao2009asymptotic} and our results here. The extension to three states is challenging enough, while the extension to more than three states is a very hard problem if one wants a strong security proof like the one we have given here. In terms of the first question, there are two possible ways to answer this. One way is that we know that at some stage as one increases the number of coherent states there must be a point where it becomes a close approximation to the full Gaussian distribution. So there may be a point where one may not need the entire (continuum) Gaussian distribution. Another way is to consider the affect that decoy state BB84 QKD has on ideal single photon BB84 and draw inspiration from there. Specifically, by increasing the number of pulses from the ideal case of one to say three pulses gives a boost to both the key rate and distance~\cite{Lo2005a,Wang2005}. So perhaps we can consider increasing the number of discrete coherent states from two to three (and potentially higher) as a decoy-state-like extension of the BPSK-modulated CV QKD protocol.
\begin{figure}[h]
   \resizebox{11cm}{!}{\includegraphics{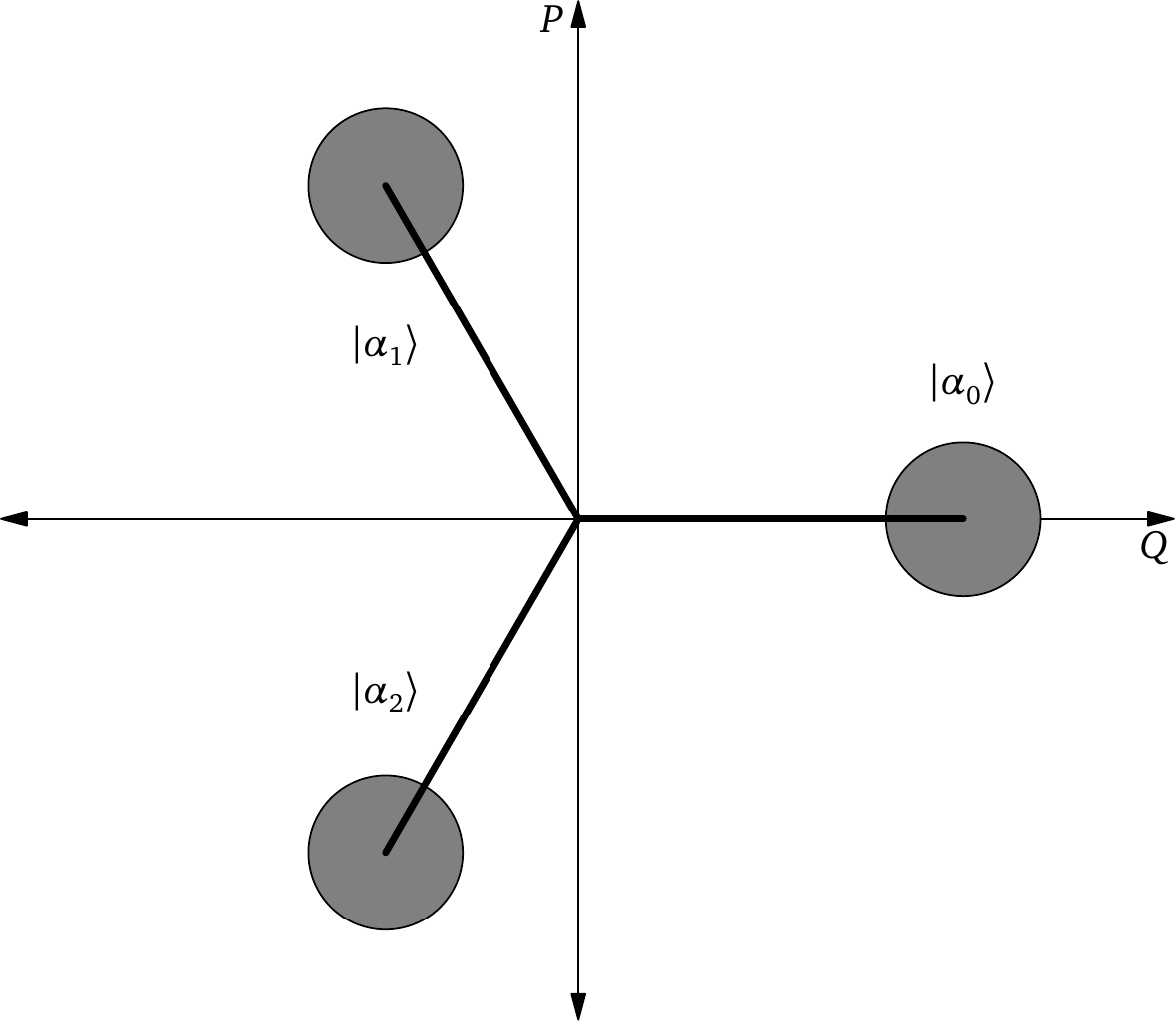}}
    \caption{Phase space configurations of the ternary coherent state QKD protocol. Note that each subsequent coherent state is $120^{\circ}$ from the other one. Alice's role is to continually and randomly choose from these three options and then send them to Bob who performs homodyne detection on the incoming states by randomly alternating between the $Q$ and $P$ quadratures. As is standard, the quantum channel is assumed to be monitored by the eavesdropper, Eve.}
    \label{fig:Phase_space}
\end{figure}

In this paper, we introduce and rigorously prove the asymptotic security of a new ternary QKD protocol based on three coherent states and homodyne detection.  For completeness, we mention that other discrete encodings for CVs have also been considered~\cite{Leverrier2009,Leverrier2010,Leverrier2011,Sych2010}. In the papers by Leverrier and Grangier~\cite{Leverrier2009,Leverrier2010,Leverrier2011}, they considered two different coherent state encodings (i.e., two and four) but in order to analyze the security they `padded out' the states with decoy-like states that effectively resembled a Gaussian distribution from Eve's point of view. This is instigated in order to leverage previous Gaussian encoding security proofs. Finally, in~\cite{Sych2010} a multi letter phase-shift keying scheme was introduced, where an $N$ number of coherent states can be used. However, the security proof only considered a lossy bosonic channel (i.e., no excess noise). In contrast, we consider a bosonic channel with arbitrary noise. Our results here allow for the significant reduction, compared to Gaussian modulation protocols, in classical post-processing, random-number generation, and classical-communication overheads. Furthermore, by keeping the benefits of CV hardware, our approach has the practical benefits of doing away with single-photon detectors that characterize DV QKD systems. Such detectors are only able to reach their promise of low-noise and high-efficiency only with the addition of cumbersome cryogenics.

\subsection*{Outline}

This paper is structured as follows. We begin by giving more background on the relationship between discrete and Gaussian encodings. This is followed by a description of the steps of our ternary coherent state protocol. Our main result is presented next and consists of a simulation of a TPSK modulated lossy bosonic channel. We end with our conclusion.

\subsection*{Notation}

In what follows, $f'$ denotes ${\mathrm{d}f\over\mathrm{d}z}$ and similarly for higher derivatives. Sometimes we explicitly mention a function's variable (typically $f=f(z)$). The symbol $\df$ stands for `defined'. The von Neumann entropy of a density matrix $\varpi_A$ is $H(A)_\varpi\equiv H(\varpi_A)\df-\Tr{}[\varpi_A\log{\varpi_A}]$~\cite{Nielsen2000,wilde2013quantum} and it becomes Shannon entropy for classical probability distributions (denoted by $X,Y$ in this paper). We will intensively study the properties of $H(X)$ where $X=\vec{x}=(x_1,x_2,1-x_1-x_2)$ and so a special name will be reserved for it -- the ternary Shannon entropy:
\begin{equation}\label{eq:ternaryEnt}
  h_3(\vec{x})\df-x_1\log{x_1}-x_2\log{x_2}-(1-x_1-x_2)\log{[1-x_1-x_2]}.
\end{equation}
The base of the logarithms is irrelevant but will be set to two throughout the paper.  The classical-quantum conditional entropy (entropy conditioned on a classical variable) reads $H(A|Y)=\sum_{y}p(y)H(\varpi_A^y)$. For a classical variable $A=X$ the entropy becomes the standard Shannon conditional entropy $H(X|Y)=-\sum_{y}p(y)\sum_xp(x|y)\log{p(x|y)}$. Other entropic quantities used in the paper include the classical mutual information $I(X:Y)\df H(X)+H(Y)-H(XY)=H(Y)-H(Y|X)$.  We will also use the quantum version of the mutual information where one of the registers is quantum and express it as $I(Y:E)=H(E|X) + I(X:E) - H(E|Y)$.

When we say a function $f$ is increasing we mean non-decreasing ($f(x)\leq f(y)$ whenever $x\leq y$). Similarly, a decreasing function means a non-increasing function.

\subsection*{Conventions}

We will use the convention of~\cite{zhao2009asymptotic} for the quadrature operators.  They  are given by $Q=1/\sqrt{2}(a+a^\dg),P=1/\sqrt{2}(a-a^\dg)$ and so $\lan(\D Q)^2\ran_\a=\lan(\D P)^2\ran_\a=1/2$ ($\hbar=1$) and $\lan Q\ran_\a=1/\sqrt{2}(\a+\bar{\a})$ where $\a=r\exp{[i\s]}$. In our case we have $(\a_x)_{x=0,1,2}$ and $\s_0=0,\s_1=2\pi/3$ and $\s_2=4\pi/3$ and $r_i=r$ is a free parameter chosen by the legitimate participants to maximize the secret key rate.

A lossy bosonic channel is a Gaussian channel parametrized by has two quantitites. One of them is the transmittance $0\leq\eta\leq1$ and the other one the number of thermal photons representing the Gaussian excess noise. For the sake of comparison, we use the definition of excess noise from~\cite{zhao2009asymptotic}:
\begin{equation}\label{eq:excessNoise}
  \d={\lan(\D Q)^2\ran_{\vr_B}\over\lan(\D Q)^2\ran_{\ket{0}}}-1
\end{equation}
given by Bob's measurement of $\vr_B$. For the simulation scenario we also assume $\lan(\D Q)^2\ran_{\vr_B}=\lan(\D P)^2\ran_{\vr_B}$. A~quantity called a ``mixedness parameter'' $\ve_x\geq0$ is upper bounded by Bob's second moments according to (65) of~\cite{zhao2009asymptotic} and it is the main estimate of the state in Eve's possession. In our simulation scenario we may set\footnote{The variables $\a,\d,\ve,\g$ used in this section should not be confused with those from Sec.~\ref{sec:main}.}, $\ve\equiv\ve_x$.

\section{Discrete versus Gaussian encoding}\label{sec:DVCVQKD}

The most studied QKD schemes are discrete-variable (DV) QKD~\cite{Scarani2009,Lo2014} and continuous-variable (CV) QKD~\cite{Weedbrook2012} based on a Gaussian encoding. The DV QKD security analysis is very mature but the secret key rates are limited given the discrete nature of the encoding. Higher-dimensional DV QKD scheme have been analyzed~\cite{bradler2016finite} but yet to  have graduated from the experimental point of view. Gaussian CV QKD offers much generous secret key rates together with a relatively simple experimental realization in terms of the state preparation and detection. But it has also its disadvantages. For instance, the classical postprocessing such as error-correction is computationally demanding and currently not very efficient. The aspiration of CV QKD based on a distribution of discrete signal states holds a promise of combining the best of both worlds.

Unlike a Gaussian encoding where the best adversary's strategy is known, the same is not true if the number of signal states is discrete. In fact, to the authors' knowledge, there exists only one paper dealing with the security of such a scheme without assuming nearly anything about the adversary's powers~\cite{zhao2009asymptotic}. The security proof (and thus the corresponding secret key rate lower bound) is derived by assuming a collective attack and in the asymptotic scenario of an infinite code length. The collective attacks are not the most general eavesdropping scheme. However, it is widely believed that similarly to DV QKD or Gaussian CV QKD, a more general attack strategy does not bring any advantage. For the second point, an asymptotic analysis is not a realistic assumption but it is historically the first step after which a finite-key length analysis typically follows. The number of signal (coherent) states prepared by a sender in~\cite{zhao2009asymptotic} is two and the receiver is allowed to measure only the first and second moments  of whatever gets through the (unknown) quantum channel. Through a tour-de-force calculation, the authors essentially construct a statistical model of the adversary's quantum states compatible with the legitimate recipient's measurement and maximize the amount of information the adversary can in principle get, following a two-way public discussion. In this way, a secret key rate lower bound is derived.

The analysis is achieved by splitting the secret key rate for a reverse reconciliation protocol into three entropic quantities and upper/lower bounding them from the quantities available from the recipient's measurement. In this paper, we follow the same strategy but instead of two signals the communicating parties exchange three coherent signals. This may seems like a small iteration but the opposite is true. We get not only substantially better secret key rate lower bounds but also show the limitation of the approach. The latter point is worth elaborating on. The proof presented in~\cite{zhao2009asymptotic} crucially relies on the monotonicity and concavity of the binary Shannon entropy as a function of the absolute value of the overlap of two pure states (not necessarily the signal states). For two signal states, these properties are trivial and they are not proved in~\cite[Eqs.~(33), (34)]{zhao2009asymptotic}. The situation dramatically changes for three signal states. Essentially, the result of this paper is the proof that these two \emph{crucial} properties hold for the ternary Shannon entropy, Eq.~\eqref{eq:ternaryEnt}. Only then can the rest of the previous analysis be applied verbatim and that is precisely what we have done. Once these two properties are proven, the rest of the proof follows exactly as in~\cite{zhao2009asymptotic} only with a few minor modifications which we will write explicitly.

There is a caveat, however. For two signal states, the binary Shannon entropy depends only on the absolute values of the overlap of the signal states. For three states, the ternary entropy depends on three possible overlaps and a certain phase. This wouldn't be a problem if we needed to study the entropy of the density matrix for the signal states only. After all, the participants are those who decide what symmetry (and a probability distribution) the signal states  obey and that could greatly simplify the analysis. The problem is that at one point of the previous analysis~\cite{zhao2009asymptotic}, the purified adversary's state (estimated from Bob's measurement) need not obey any such property and the state must be considered arbitrary. As it is discussed in the first remark of Section~\ref{subsec:ternaryDM}, in the presence of more than one overlap, the studied function does not even satisfy  the (suitably generalized) notion of monotonicity. This is not only surprising but it also affects the applicability of the approach of~\cite{zhao2009asymptotic} that we follow here -- unlike the case of two signal states, the proof strategy has its limits. Another consequence of our generalization is that unless a generic argument for monotonicity and concavity of the suitable generalized entropy function can be found (taking into account what we have just stated), it is most likely that a completely different approach is needed in order to study discrete CV QKD protocols and their rates for more than three signal states.

\section{Description of Ternary Coherent State Protocol}\label{sec:privateCode}

Here we outline our ternary (three coherent state) QKD protocol. It goes as follows.
\begin{enumerate}
\item Alice prepares one of three possible coherent states $\ket{\a_i}$ with probability $p_i=1/3$, where $i = 0,1,2$. In Fig.~\ref{fig:Phase_space}, we have a schematic of the phase space depicting how the three coherent states are placed, i.e., sequentially separated by $120^{\circ}$. She then  sends the randomly selected coherent state to the receiver, Bob, over an insecure quantum channel. It is assumed that this channel could be monitored by Eve. Alice repeats this step many times. Alice's choice for the $i$th signal (coherent state pulse) is recorded in the variable $x_i$. Specifically, the labeling goes as: $\ket{\a_0}$ is $x_i = 0$, $\ket{\a_1}$ is $x_i = 1$, and $\ket{\a_2}$ is $x_i = 2$.

\item Bob, upon receiving a sequence of quantum states, randomly performs homodyne detection thereby randomly measuring the quadratures $Q(\phi)$ for $\phi=(\pi/2,-\pi/6,-5\pi/6)$ of each of the coherent states. A similar setup was used in~\cite{namiki2006efficient} but tested on a specific eavesdropping strategy. Bob's measurement results are recorded in the variable $y_i$. Note that $Q(\pi/2)\equiv P$ in Fig.~\ref{fig:Phase_space}.

\item After the transmission, the parties publicly announce the measurement quadratures. One of the quadratures, say $Q(-5pi/6)$, the measurement data is published which is used to determine the extent of the adversary's maliciousness. These data are subsequently discarded.

\item The remaining data (which we denote as $\{\vec{x},\vec{y}\}$) will be used for the final key generation. For the purpose of reverse reconciliation, Bob sends computes  functions $u(\vec{y})$ and $w(\vec{y})$ and sends $u(\vec{y})$ over a public channel to Alice and keeps $w(\vec{y})$ which is a discrete proto-key (partially correlated with Alice's discrete variable $\{\vec{x}$).

\item Classical post-processing procedures of error correction and privacy amplification are applied by Alice and Bob in order to extract the final shared secret-key. This final secret bit string is then used as a one-time pad in order to perfectly secure messages.
\end{enumerate}

\section{A secret key rate lower bound }

In this section, we derive the lower secret key rate for the ternary protocol with respect to a lossy bosonic channel. Mathematically the main results needed for this lower bound (and which are rigorously proven in the Appendix) involve proving that monoticity and concavity both hold for the ternary Shannon entropy, Eq.~\eqref{eq:ternaryEnt}. We begin by defining the lower bound of the secret key rate $K$ followed by calculating the individual components of this bound which include Alice and Bob's mutual information and Eve's mutual information.

The secret key rate $K$ is lower bounded as
\begin{equation}\label{eq:rateRR}
 K > I(X:Y) - \max\limits_{\vr_{ABE}} I(Y:E)
\end{equation}
Eq.~\eqref{eq:rateRR} has its origin in~\cite{devetak2005private} where the one-way private quantum channel capacity was established. The lower bound also differs from~\cite{devetak2005private} in several aspects. (i) The channel is a priori not known and is only partially estimated by the measurements of the legitimate participants. The ambiguity in its identification is an advantage for Eve -- the optimization leads to the penalty on the amount of shared secret correlations as if Eve used the best eavesdropping channel compatible with the measurements. This translates into the best channel purification $\vr_{ABE}$ held by Eve among all admissible ones in Eq.~\eqref{eq:rateRR}, see also Ref.~\cite{devetak2005distillation}. (ii) Our key distribution protocol uses reverse reconciliation where the classical communication (exploited by Eve) is transmitted from Bob to Alice. This results in the appearance of the second term in~\eqref{eq:rateRR} as opposed to~\cite{devetak2005distillation,devetak2005private} dealing with direct reconciliation. (iii) Finally, given the reality of the explicit quantum private code described in Sec.~\ref{sec:privateCode}, the RHS of~\eqref{eq:rateRR} is a one-shot formula -- a natural lower bound to a multi-letter secret key rate formula. A closely related expression for a secret key rate was derived in~\cite{kraus2007security} while focusing solely on the security of QKD.

\subsection{A secret key rate lower bound for a Lossy Bosonic Channel}\label{subsec:simLossy}

The job here is to maximize the mutual information $I(Y:E)$ in order to find a lower bound on the secret key rate~$K$. In an actual experiment, the classical probability distribution must be measured to be subsequently inserted to the relevant entropic quantities in~\eqref{eq:rateRR}. Following~\cite{zhao2009asymptotic} we may simulate an actual link by a lossy bosonic channel. This is a realistic model for the atmospheric CV QKD with homodyne measurement. Note that the complementary channel is another lossy bosonic channel and it captures the effect of the environment or an adversary Eve. As is common for QKD, Eve is assumed to control the channel and take an advantage of the generated noise to hide her illicit behavior.

As we will see in Section~\ref{subsec:ternaryDM}, unlike the BPSK case studied in~\cite{zhao2009asymptotic} the entropic properties of the investigated density matrix depend not only on the mutual overlaps of the three signal states but also on the overall phase, see the expressions for $d$ in Eq.~\eqref{eq:abcdGen} or~\eqref{eq:cdSpecCased}. In the simulation scenario for a lossy bosonic channel the phase can be computed as we will show now.

We will first consider the zero excess noise case $\d={\lan(\D Q)^2\ran_\a\over\lan(\D Q)^2\ran_{\ket{0}}}-1=0$ (a pure-loss bosonic channel). The estimated quantities become simpler as the recipient's detected states are pure coherent states and similarly for Eve. The parameter $\ve$ given by (65) in~\cite{zhao2009asymptotic} is bounded from above by $U\equiv U_x=0$ from~(65). Hence $\ve=0$ and (66) together with (C17,C18) of~\cite{zhao2009asymptotic} imply
$$
|\brk{\tilde{\b}_i}{\tilde{\b}_j}|=c_u=c_l=\ka.
$$
The RHS is given by $\ka\equiv\ka_{ij}=|\brk{\sqrt{\eta}\a_i}{\sqrt{\eta}\a_j}|$. Inserting $c_u,c_l$ into (70,71) in~\cite{zhao2009asymptotic} we get
\begin{equation}\label{eq:lambda}
d_l=d_u={|\brk{\a_i}{\a_j}|\over\ka}\df|\g_{ij}|\equiv|\g|=e^{-{3\over2}(1-\eta)r^2}.
\end{equation}
This quantity is the estimated overlap of the states going to the environment. As expected from the properties of a pure-loss bosonic channel it is the same quantity as $\ka$ with $\eta$ substituted by $1-\eta$.

We can geometrically interpret the product of inner products in~\eqref{eq:abcdGen} (or its special case~\eqref{eq:cdSpecCased}) if $\psi_i$ are coherent states. Then the product
\begin{equation}\label{eq:innerProducts}
z_{01}z_{12}z_{20}=\brk{\a_0}{\a_1}\brk{\a_1}{\a_2}\brk{\a_2}{\a_0}=e^{-{1\over2}(c_{01}^2+c_{12}^2+c_{20}^2)}e^{-i2(A_{01}+A_{12}+A_{20})}
\end{equation}
is written in terms of the sides $c_{ij}$ and area $A_{012}\df A_{01}+A_{12}+A_{20}$ of the triangle formed by the corresponding three points in phase space. This is the interpretation provided by~Lemma~\ref{lem:phase}.

We illustrate it on the symmetric case $c_{01}=c_{20}=c_{12}\equiv c$ of an equilateral triangle for $\d=0$, whose side squared is equal to $c^2=3r^2(1-\eta)$ found in~\eqref{eq:lambda}.  From the new triangle side we deduce, with the help of elementary geometry (essentially Heron's formula), the corresponding  area:
\begin{equation}\label{eq:heron}
  A_{012}={1\over4}\big(4c_{01}^2c_{12}^2-(c_{01}^2+c_{12}^2-c_{20}^2)^2\big)^{1/2}.
\end{equation}
and consequently the phase: $\vt=2A_{012}=r^2{3\sqrt{3}\over2}(1-\eta)$.

How do we apply it to the $\d>0$ case? Here, the situation is  slightly different. The effect of a lossy bosonic channel is not only  shrinking of the phase space triangle but also increasing the states' variances -- environment (Eve) and Bob do not receive a mixture of three pure states but rather of three mixed Gaussian states. Following the general procedure outlined in~\cite{zhao2009asymptotic}, where only the first and second moments are measured, the overlaps of Eve's state figuring in our simulation scenario are bounded by (70) and~(71) in ~\cite{zhao2009asymptotic}. In that case, neither $|\g|$ nor $\ka$ are overlaps of the corresponding pure coherent states. More precisely, since Bob measures only the first two moments, the authors of~\cite{zhao2009asymptotic} introduced fiducial coherent states $\ket{\overline{\b}_i}$ on Bob's side compatible with the measurement of the first moment. Then $\ka=|\brk{\overline{\b}_i}{\overline{\b}_j}|$  and  as before $\ka\equiv\ka_{ij}=|\brk{\sqrt{\eta}\a_i}{\sqrt{\eta}\a_j}|$ for the case of a lossy bosonic channel\footnote{An insight provided by Saikat Guha.}. This provides the same interpretation for $|\g|$ (Eve's parameters estimated from Bob's measurement) and the phase is then  determined according to Lemma~\ref{lem:phase}.

The main object of study is a lower bound on the secret key rate, Eq.~\eqref{eq:rateRR}. Here we break down the lower bound for the simulated lossy bosonic channel. The central role is played by the ternary Shannon entropy, Eq.~\eqref{eq:ternaryEnt},  where $x_k=t_k+1/3$ and $t_k$ is given by~\eqref{eq:tkTrig}.

\subsection*{Eve's and Alice's Mutual Information,  \boldmath{$I(X:E)$} }
Closely following~\cite[Sec. IV.~B]{zhao2009asymptotic}, to get a secret key lower bound, the first quantity to estimate is $I(X:E)<I(X:QE)=h_3(\vec{x}(Z))$ for $x_k$ restricted to $p_k=1/3$ and $\brk{\Psi_{EQ}^i}{\Psi_{EQ}^j}=Z_{ij}={Z}\exp[{i\tilde\tau_{ij}}],\,Z>0$. As explained in the remark on p.~\pageref{rem:restrictedsolution}, the restriction to $|Z_{ij}|=Z$ is a necessary step for the proof strategy following~\cite{zhao2009asymptotic} to go through. Then, from~\eqref{eq:tkTrig}, we get the explicit form of $x_k$:
\begin{subequations}\label{eq:xkForpkOneThird}
\begin{align}
  x_1 & ={1\over3}\Big(1+2Z\cos{\vt\over3}\Big), \\
  x_{2,3} & ={1\over3}\Big(1-Z\big(\cos{\vt\over3}\mp\sqrt{3}\sin{\vt\over3}\big)\Big).
\end{align}
\end{subequations}
Denoting $f\equiv f_{ij}=F(\vr_E^i,\vr_E^j)$ to be the fidelity of $\vr^{i(j)}_E=\Tr{Q}[\Psi^{i(j)}_{EQ}]$ we get
\begin{equation}\label{eq:h3UpBndMutInfo}
h_3(\vec{x}(Z,\vt))\leq h_3(\vec{x}(f,\vt))\leq h_3\big(\vec{x}((1-\tilde\ve_0)^{1/2}(1-\tilde\ve_1)^{1/2}|\g|,\vt)\big)
\end{equation}
where $0\leq\tilde\ve_i\leq\ve$. The second inequality follows from the proof of monotonicity, Theorem~\ref{thm:entrDecreasing}, as a special case $p_k=1/3$.

When restricted to the simulation scenario of a lossy bosonic channel, the parameter $\vt$ is a phase whose value we determine with the help of Lemma~\ref{lem:phase}. Before doing so, recall that for $\d=0$ the lossy bosonic channel merely ``shrinks'' the triangle representing the mixture of three coherent states in phase space and the shrinking factor is $1-\eta$ for Eve's system (see~\eqref{eq:lambda}). Consequently, $\vr_E^i$ are pure and Eq.~\eqref{eq:lambda} can be interpreted as the modulus of their overlap.

\subsection*{Eve's Entropy conditioned on Alice's variable  \boldmath{$X,\,H(E|X)$}}

The next expression used for the secret key estimation is the conditional entropy $H(E|X)$. It is upper  bounded by~\cite{zhao2009asymptotic}
$$
{1\over3}\sum_{x}{(1+V_x)\log{[1+V_x]}-V_x\log{V_x}},
$$
where $V_x=\big(\lan(\D Q)^2\ran_{\vr_B}\lan(\D P)^2\ran_{\vr_B}\big)^{1/2}-1/2$. In the case of a lossy bosonic channel we find $V_x=\d/2$.

\subsection*{Eve's Entropy conditioned on Bob's measurement outcome  \boldmath{$Y,\,H(E|Y)$}}

The third expression needed to be evaluated from the secret key lower bound is $H(E|Y)$ in (62) from~\cite{zhao2009asymptotic}. In order to do so we have to generalize the conditional probability distribution related to the action of a lossy bosonic channel. We cannot simply take the derived expressions in~\cite{zhao2009asymptotic} since for three and more signal states the states cannot all be aligned with a real line in phase space. Instead, we introduce
$$
p(y|x)={1\over\pi(1+\d)}\exp{\bigg[-{|y-\sqrt{\eta}\a_x|^2\over\d+1}\bigg]}
={1\over\pi(1+\d)}\exp{\bigg[-{|y|^2+\eta r^2-2|y|r\sqrt{\eta}\cos{[\phi-\s_x]}\over\d+1}\bigg]},
$$
where $y=|y|\exp{[i\phi]}$ and $\a_x=r\exp{[i\s_x]}$. For three signal states we take the values of $\s_{0,1,2}$ introduced in Section~\ref{sec:privateCode}. To simulate the channel we further use $p(x|y)={1\over3}p(y|x)/p(y)$ together with
$$
p(y)=\sum_{x=0,1,2}p(y|x)p(x)={1\over3}{1\over\pi(1+\d)}\sum_{x=0,1,2}\exp{\bigg[-{|y-\sqrt{\eta}\a_x|^2\over\d+1}\bigg]}.
$$
Hence, for example,
$$
p(0|y)={\exp{\Big[-{|y-\sqrt{\eta}\a_0|^2\over\d+1}\Big]}\over
    \sum_{x=0,1,2}\limits\exp{\Big[-{|y-\sqrt{\eta}\a_x|^2\over\d+1}\Big]}}.
$$
A straightforward generalization of the derivation of Eqs. (56) and (57) in~\cite{zhao2009asymptotic} allows us to lower bound $H(E|Y)$.

\subsection*{Alice's and Bob's Mutual Information, \boldmath{$I(X:Y)$} }

The final component is the classical mutual information $I(X:Y)=H(X)-H(X|Y)$ calculated with the help of $p(x|y)$ and $p(y)$ defined above.

\subsection*{Final Secret Key Rate Lower Bound for a Lossy Bosonic Channel}

Now we have all the ingredients we need to find the actual secret key rate lower bound.  It is expression~(72) given in~\cite{zhao2009asymptotic}, adapted to the TPSK encoding. It can be written as
\begin{align}\label{eq:Ratedelta}
  K & > \underbrace{\log{3}-\int_{0}^\infty\dif{|y|}|y|\int_0^{2\pi}\dif{\phi}p(y)\sum_{x=0,1,2}p(x|y)\log{[p(x|y)]}}_{I(X:Y)}\nn\\
     &-\underbrace{\big((1+\d/2)\log{[1+\d/2]}-\d/2\log{[\d/2]}\big)}_{H(E|X)}-\max_{0\leq\tilde\ve\leq\ve}
     \Bigg[\underbrace{h_3\big(\vec{x}((1-\tilde\ve)|\g|,\vt)\big)}_{H(X:E)}\\
     &\left.\begin{array}{@{}l}\nn
    {\displaystyle- \int_{0}^\infty\dif{|y|}|y|\int_0^{2\pi}\dif{\phi}p(y)h_3(\vec{x}(|\g|,\vt,p(0|y),p(1|y)))}\\
    {\displaystyle+\sum_{x=0,1}\Bigg[\bigg({\tilde\ve\over3}{1+|\g|\over1-|\g|}\bigg)^{1/2}\Bigg(\int_{0}^\infty\dif{|y|}|y|\int_0^{2\pi}\dif{\phi}
  p(y){h_3^2\big(\vec{x}(|\g|,\vt,p(0|y),p(1|y))\big)\over p(x|y)}\Bigg)^{1/2}\Bigg]}\\
  {\displaystyle+{\tilde\ve\over1-|\g|}h_3\big(\vec{x}(|\g|,\vt,{1/3},1/3)\big)}
   \end{array}\!\!\right\}&\!\!\!-H(E|Y).
\end{align}
For ease of sight we identified the \emph{origin} of the summands by the expressions in the braces.
\begin{figure}[h]
   \resizebox{14cm}{!}{\includegraphics{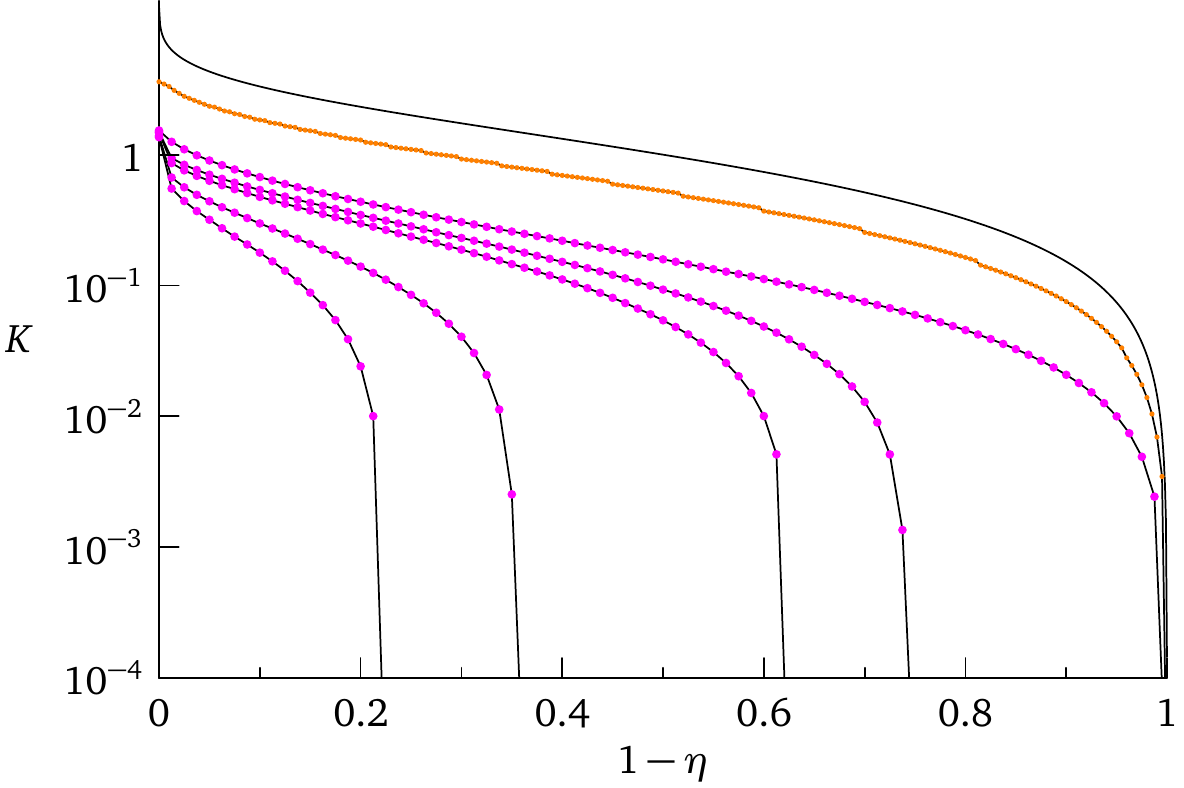}}
    \caption{Secret key rates as functions of loss $1-\eta$ for several values of the channel excess noise $\d=(0,0.0004,0.001,0.005,0.01)$ (the pink dots). The black curve is the ultimate achievable bound without an energy constraint for $\d=0$. The orange curve is an achievable bound for $\d=0$ taking into account the input energy constraint~\cite{pirandola2015fundamental}. All curves are functions of the channel loss.}
    \label{fig:rates}
\end{figure}
The main technical result of this paper -- the proofs of monotonicity and concavity of the ternary Shannon  entropy -- participate in the derivation of $H(E|Y)$. The reasoning is nearly a verbatim copy of Section~IV.C and the Appendices~A and~C of~\cite{zhao2009asymptotic} implying the conditional entropy to be a lower bound on the secret key rate~$K$.

In Fig.~\ref{fig:rates} we present the main result of our analysis (applied to a simulated lossy bosonic channel). We plot the secret key lower bound, Eq.~\eqref{eq:Ratedelta}, for several values of the excess noise parameter. Compared to~\cite{zhao2009asymptotic}, we find better lower bounds as expected from the use of three signals states but also much better threshold values where the rate is zero. It therefore supports the idea that to approach the high rates given by a continuous Gaussian encoding, one would need only a reasonably small number of signal states. This cannot, strictly speaking, be correct for the vicinity of $\eta=1$. It is known  that the ultimate upper bound for the two-way secret key rate at the presence of zero excess noise is equal to $K=-\log{[1-\eta]}$~\cite{pirandola2015fundamental}, a quantity diverging for $\eta\to1$. Clearly, for any finite number of discrete signal states $d$, the maximal secret key rate for $\eta=1$ is $\log{d}$ like in our case $d=3$. Ref.~\cite{pirandola2015fundamental} also provided an achievable bound (actually a lower bound based on~\cite{Garcia2009})  by taking into account the input energy constraint. This is depicted in Fig.~\ref{fig:rates} as the orange dotted curve for $\d=0$. The `stairs' on this curve are the consequence of a different optimal energy (input state overlap leading to a different input energy constraint) shown in Fig.~\ref{fig:ropti}.

An important fact to realize is that even though we have only proved monotonicity and concavity of $h_3$ for $0\leq\vt\leq\pi/2\Leftrightarrow\ve\leq0$ (for $\ve$ given by~\eqref{eq:pqEpsi}), it does not affect the secret key rate lower bound. The optimal input energy falls inside the region $\ve\leq0$. The situation is also depicted in Fig.~\ref{fig:ropti}.
\begin{figure}[h]
   \resizebox{13cm}{!}{\includegraphics{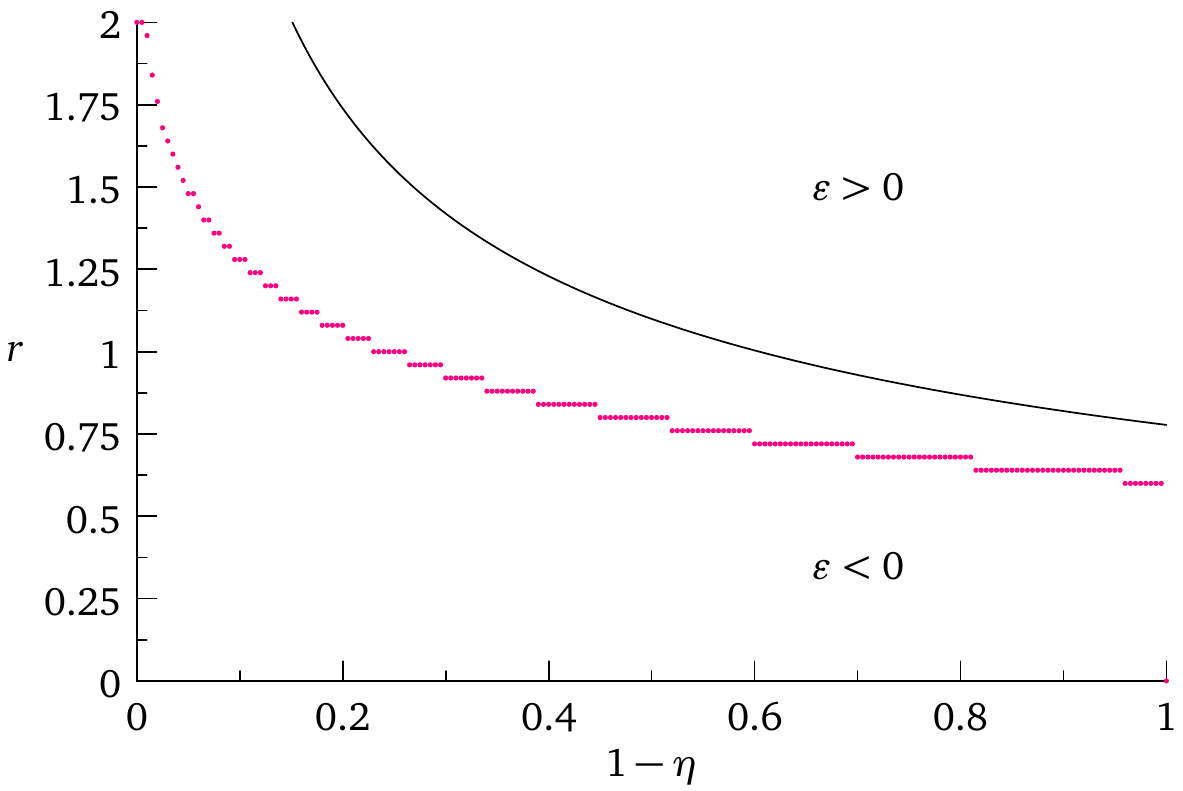}}
    \caption{The red dots depict the optimizing overlaps $r$ for $\d=0$. The black curve is a boundary $\ve=0$ of~\eqref{eq:pqEpsi} ($\vt=\pi/2$) given by $r^2{3\sqrt{3}\over2}\eta=\pi/2$ (see below Eq.~\eqref{eq:heron}) below which the proofs of monotonicity and concavity exist ($\ve\leq0$).}
    \label{fig:ropti}
\end{figure}

\section{Differences in an actual QKD experiment}

The real-world scenario introduces further complications. The channel may not be lossy bosonic (it may not even be described by a stationary process for the duration of the experiment but we will avoid this type of complications). For a stationary channel and in the asymptotic scenario the participants collect enough statistics to reconstruct the channel to estimate the conditional probability distributions $p(y|x)$ and $p(y)$ arbitrarily well. The same applies to the BPSK analysis from~\cite{zhao2009asymptotic} but as we already alluded to, there is more degrees of freedom in the ternary case. There are in total three overlaps in the form of three real parameters for a general triple of coherent pure states and in addition there is a phase. In the simulation scenario of a lossy bosonic channel the overlaps if chosen symmetrically by Alice (our assumption) and the phase can be subsequently calculated as done in the previous section\footnote{Note that similarly to~\cite{zhao2009asymptotic} we not only calculate the entropy of the input density matrix but also of other, say intermediate, density matrices in order to lower bound the secret key rate. Even there the three real parameters coincide (they can't be interpreted as overlaps, though, see below Eq.~\eqref{eq:heron}) and the phase can be calculated for a lossy bosonic channel.} But in for an actual experiment we can only assume the symmetry of a density matrices directly prepared by Alice. The states where Eve can in principle intervene has no a priori symmetry which translates into their entropy to be dependent on three plus one free parameters. As it turns out (see the discussion in Sec.~\ref{subsec:ternaryDM}), the key property of monotonicity of the ternary Shannon entropy does not hold in general and  the strategy to lower bound the secret key rate from~\cite{zhao2009asymptotic} must be abandoned.

How do we overcome this problem here? If the parameters measured by Bob indicate that the incoming states are not symmetrically distributed, the participants assume the closest symmetric distribution that gives Eve the biggest advantage. One could be tempted to take the smallest of the three overlaps and create a symmetric distribution based on it. However, as the example in~\cite[p.~10]{jozsa2000distinguishability} shows, the entropy of such a density matrix does not necessarily becomes smaller thus indicating \emph{more} distinguishable quantum states. So a better strategy to introduce a single overlap is called for and it will necessarily reduce the secret key rate. But only this is the situation for which we can follow the proof in~\cite{zhao2009asymptotic} once the monotonicity and concavity of the ternary Shannon entropy is proven. The worst case scenario happens if Bob detects only two states, that is, if the channel is so disruptive that it managed to merge two signal states to one quantum state. In that case the secret key rate would be zero and it would probably be better to switch to BPSK.

How do we recover the other free parameter, namely the angle? Similarly to the lossy bosonic case, a triple of fiducial coherent states $(\ket{\overline{\b}_i})_{i=0,1,2}$ with the same absolute value of the overlap is introduced. We assume that the triple properly bounds the entropies as described in the previous paragraph, so that the advantage is given to Eve resulting in the key rate reduction. Then we  followed the procedure of phase calculation described below Eq.~\eqref{eq:heron} following Lemma~\ref{lem:phase}. This is the right phase for the fiducial triple of pure coherent states.

 \section{Conclusion}

In conclusion, we introduced and rigorously proved the asymptotic security of a new ternary QKD protocol based on three coherent states and homodyne detection. The motivation for introducing such a protocol is to extract a best-of-both-world's approach to QKD in terms of the encoding and decoding of discrete variable schemes along with the practical hardware of continuous variable schemes. There is, however, the downside that the security proof is very challenging compared to the results for Gaussian modulated continuous-variable QKD protocols. We overcame this challenge by mathematically proving that two crucial properties, monotonicity and concavity, hold for the ternary Shannon entropy. This allowed us to evaluated a lower bound to the secret key rate in the collective attack scenario.

Other interesting avenues of research could include considering a four-state extension (if possible, or perhaps using a different method), determining what number of signal states are enough to tend close to the full Gaussian distribution and also a thorough finite-key analysis. This is a lively area of research for many classes of bosonic channels where the lossy bosonic channel is an important subclass~\cite{pirandola2015fundamental,laurenza2017finite}. A measurement-device-independent (MDI)-QKD~\cite{Braunstein2012,Lo2012,pirandola2015high} version of our scheme presented here would also be interesting as a way of ruling out side channel attacks.

\section*{Acknowledgement}
We would like to thank Saikat Guha for helpful discussions. The authors  acknowledge support from the U.S. Office of Naval Research (ONR). This material is based upon work supported by the Air Force Office of Scientific Research under award number FA9550-17-1-0083. The authors thank Saikat Guha for valuable comments and discussions.

\appendix
\renewcommand{\thesubsection}{A.\arabic{subsection}}
\section{Full Details of Main result}\label{sec:main}

\subsection{Properties of ternary density matrix}\label{subsec:ternaryDM}

In this section, we give the calculations needed to prove the main results. To begin with, let
\begin{equation}\label{eq:Adm}
  \varpi=p_0\kbr{\psi_0}{\psi_0}+p_1\kbr{\psi_1}{\psi_1}+p_2\kbr{\psi_2}{\psi_2}
\end{equation}
be a rank-three density operator where $p_0+p_1+p_2=1$. The state $\varpi$ takes on a different meaning depending on where it is used. It can be an input density matrix a sender prepares in a lab in which case $p_k=1/3$ and $\psi_k$ are the signal (coherent) states with a chosen symmetry. Or, it can be Eve's conditioned state based on Bob's measurement. In that case, $p_k$ are arbitrary conditional probabilities $p_k(x|y)$  and $\psi_k$ are  pure states with no obvious symmetry properties~\cite{zhao2009asymptotic}.

Following the Cayley-Hamilton theorem, one finds the coefficients of the characteristic polynomial
\begin{equation}\label{eq:CH}
  \det{[\varpi-x\id]}=f(x)=ax^3+bx^2+cx+d=0,
\end{equation}
where
\begin{subequations}\label{eq:abcdFromCH}
  \begin{align}
    a & =1, \\
    b & =-\Tr{}[\varpi]=-1,\\
    c & ={1\over2}((\Tr{}[\varpi])^2-\Tr{}[\varpi^2])={1\over2}(1-\Tr{}[\varpi^2]),\\
    d & =-{1\over6}((\Tr{}[\varpi])^3-{3}\Tr{}[\varpi]\,\Tr{}[\varpi^2]+2\Tr{}[\varpi^3])=-{1\over6}(1-3\Tr{}[\varpi^2]+2\Tr{}[\varpi^3]).
  \end{align}
\end{subequations}
The last  two coefficient become
\begin{subequations}\label{eq:abcdGen}
  \begin{align}
    c & ={1\over2}(1-p_0^2-p_1^2-p_2^2-2p_0p_1|z_{01}|^2-2p_1p_2|z_{12}|^2-2p_0p_2|z_{02}|^2),\\
    d & =\frac{1}{6} \Big(-1+3 \left( p_0^2+  p_1^2+p_2^2+2  {p_0}  {p_1} |z_{01}|^2+2  {p_0}  {p_2} |z_{02}|^2+2  {p_1}  {p_2} |z_{12}|^2\right)\nn\\
    &\quad-2 \big( p_0^3+p_1^3+p_2^3+3  (p_0^2  p_1+p_0p_1^2)   |z_{01}|^2+3  (p_0^2  {p_2}+{p_0}p_2^2)|z_{02}|^2+3(p_1^2{p_2}+{p_1}p_2^2)|z_{12}|^2\nn\\
    &\quad+3{p_0}{p_1}p_2(z_{01}z_{12}z_{20}+c.c)
    \big)\Big).
  \end{align}
\end{subequations}
Note that $\varpi$ in all its roles in the security proof if always a sum of rank-one operators. Hence the trace quantities in Eqs.~\eqref{eq:abcdFromCH} are easy to find. An additional check was performed by calculating the quartic term
$$
{1\over24}\big((\Tr{}[\varpi])^4-6(\Tr{}[\varpi])^2\Tr{}[\varpi^2]+3(\Tr{}[\varpi^2])^2  + 8\Tr{}[\varpi]\Tr{}[\varpi^3]-6\Tr{}[\varpi^4]\big)
$$
and was found to be zero as it should be.

We set the overlaps to be $\brk{\psi_i}{\psi_j}=z_{ij}=|z|\exp[{i\tau_{ij}}]$ and get
\begin{subequations}\label{eq:cdSpecCase}
  \begin{align}
    c & ={1\over2}\big(1-p_0^2-p_1^2-p_2^2-|z|^2(2p_0p_1+2p_1p_2+2p_0p_2)\big),\\
    d & =\frac{1}{6} \Big(-1+3 \left( p_0^2+  p_1^2+p_2^2+2|z|^2(  {p_0}  {p_1} + {p_0}  {p_2} + {p_1}  {p_2})\right)\nn\\
    &\quad-2 \big( p_0^3+p_1^3+p_2^3+3  (p_0^2  p_1+p_0p_1^2)   |z|^2+3  (p_0^2  {p_2}+{p_0}p_2^2)|z|^2+3(p_1^2{p_2}+{p_1}p_2^2)|z|^2\nn\label{eq:cdSpecCased}\\
    &\quad+6|z|^3{p_0}{p_1}p_2\cos{\vt} \big)\Big),
  \end{align}
\end{subequations}
where $\vt=\tau_{01}+\tau_{12}+\tau_{20}$.  The absolute value $|z|$ and the angle $0\leq\vt\leq\pi$ are not independent and we will revisit the relation below Eq.~\eqref{eq:tkTrig} (see also Lemma~\ref{lem:phase}).
\begin{rem}\label{rem:restrictedsolution}
  It may seem that by setting $|z_{ij}|=|z|,\forall i,j$ we limit ourselves to a special case of~$\varpi$. This is indeed true. Quite surprisingly, however, it is the most general case for which one of the studied properties (monotonicity) actually holds. It turns out that the multivariable function studied in this paper, the ternary Shannon entropy (Eq.~\eqref{eq:ternaryEnt}), is not monotone decreasing unless $|z_{ij}|=|z|,\forall i,j$ in which case it reduces to the standard single-variable problem.  What does it mean for a multivariable function to be monotone increasing/decreasing? This question is closely related to the existence of sets that cannot be totally ordered (totality means  that either $x\leq y$ or $y\geq x$ holds). An example is~$\bbR^n$ for $n>1$ which is only a partially ordered set. To this end, one defines the componentwise order~\cite{ghorpade2010course} of two $n$-tuples $(x_1,\dots,x_n)\leq(y_1,\dots,y_n)$ iff $x_i\leq y_i,\forall i$. A monotone increasing or decreasing function $f:\bbR^n\mapsto\bbR^m$ then satisfies $f(x_1,\dots,x_n)\leq f(x_1,\dots,x_n)$ and $f(x_1,\dots,x_n)\geq f(x_1,\dots,x_n)$, respectively.  The lack of this property (namely not decreasing) means that the strategy outlined in~\cite{zhao2009asymptotic} we follow here is simply not applicable.
\end{rem}

Coefficients, Eqs.~\eqref{eq:cdSpecCase}, are used to get the eigenvalues of~$\varpi$.  Following~\cite{waerden1966vol} (or~\href{https://en.wikipedia.org/wiki/Cubic_function}{Wikipedia} for a quick summary) we form
\begin{subequations}
  \begin{align}
    \D_0 & = b^2 - 3 a c=1 - 3 c,\\
    \D_1 & = 2 b^3 - 9 a b c + 27 a^2 d=-2 + 9 c + 27 d
  \end{align}
\end{subequations}
and define
\begin{subequations}\label{eq:pq}
  \begin{align}%
    p & =-{\D_0\over3}=\a+\b z^2,\label{eq:pFcn} \\
    q & ={\D_1\over27}=\g+\d z^2+\ve z^3\label{eq:qFcn}.
  \end{align}
\end{subequations}
They are the coefficients of a reduced cubic $t^3+pt+q$ the general cubic polynomial $f(x)$ can be converted to. The coefficients of $p,q$ from Eqs.~\ref{eq:pq} are given by
\begin{subequations}\label{eq:pqCoeffs}
  \begin{align}
    \a & = {1\over6}(1-3 p_0^2-3p_1^2-3p_2^2)\leq0,\label{eq:pqAlpha}\\
    \b & = -(p_0p_1+p_0p_2+p_1p_2)\leq0,\label{eq:pqBeta}\\
    \g & = {1\over27}(-2 + 9 p_0^2 - 9 p_0^3 + 9 p_1^2 - 9 p_1^3 + 9 p_2^2 - 9 p_2^3)\nn\\
       & = {1\over27}(3p_1-1)(3p_2-1)(3p_1+3p_2-2)\lessgtr0,\label{eq:pqGamma}\\
    \d & = {1\over27} \big(18 (p_0 p_1+p_0 p_2+ p_1p_2)- 27 (p_0^2 p_1+p_0 p_1^2+p_0^2 p_2+p_1^2 p_2+p_0p_2^2+p_1p_2^2)\big)\leq0 , \label{eq:pqDelta}\\
    \ve & = -2p_0p_1p_2\cos{\vt}\lessgtr0,\label{eq:pqEpsi}
  \end{align}
\end{subequations}
where we also summarized some basic properties based on $0\leq p_i\leq1,\sum_ip_i=1$. Then, the three roots (the eigenvalues of $\varpi$) are $x_k=t_k-b/(3a)=t_k+1/3$ where
\begin{equation}\label{eq:tkTrig}
    t_k=2\sqrt{-{p\over3}}\cos{\bigg({1\over3}\arccos{\bigg({3\over2}{q\over p}\sqrt{-{3\over p}}\bigg)}-{2k\pi\over3}\bigg)}.
\end{equation}
It is known~\cite{waerden1966vol} that
\begin{align}\label{eq:tkPropsSumZero}
  t_0+t_1+t_2 &= 0, \\
  t_0\geq t_1&\geq t_2\label{eq:tkPropsOrdered}
\end{align}
hold. Hence $x_0+x_1+x_2=1$ as we expect from $\Tr{}[\varpi]=1$ but $x_2\geq0$ is not satisfied for all $|z|$ and $\vt$. For example, if $\psi_1=e^{i\vp_1}\psi_0,\psi_2=e^{i\vp_2}\psi_0$ then $|z|=1$ and $\vt=\tau_{01}+\tau_{12}+\tau_{20}=0$. In general, it  turns out that  $x_2\geq0$ is equivalent to $q\leq{1\over27}+{p\over3}$ which provides a bound on $\vt$ given $|z|$. Indeed, for $|z|=1$ the only possibility is $\vt=0$.

Something much stronger can be said about the phases if $\psi_i$ are actual coherent states (either the signal states or the fiducial states we mentioned in the main text).
\begin{lem}\label{lem:phase}
  The  phase $\mathrm{Arg[\brk{\a_i}{\a_j}]}$ of an inner product of two coherent states $\ket{\a_i}$ and $\ket{\a_j}$ is a function of $|\brk{\a_i}{\a_j}|$.
\end{lem}
\begin{proof}
  Using elementary trigonometry we write
  \begin{equation}
    \brk{\a_i}{\a_j}=e^{-{1\over2}(r_i^2+r_j^2-2r_ir_j\cos{[\s_j-\s_i]})}e^{-i2r_ir_j\sin{[\s_j-\s_i]}}=e^{-{1\over2}c_{ij}^2}e^{-i2A_{ij}},
  \end{equation}
  where $c_{ij}$ a side of triangle opposite to the angle $\tau_{ji}=\s_j-\s_i$ between the sides $r_i$ and $r_j$ and $A_{ij}$ is the triangle area (it is oriented since $A_{ij}=-A_{ji}$). But knowing $r_i,r_j,c_{ij}$, we can easily calculate the area of the triangle and hence the phase $\mathrm{Arg{\,[\brk{\a_i}{\a_j}]}}=-2A_{ij}=2A_{ji}$. Hence the phase is much more constrained if $\psi_i$ are coherent states.
\end{proof}
This trivial statement (we could also use the relation between $\sin{}$ and $\cos{}$ to get the phase) has interesting consequences we exploited in Eq.~\eqref{eq:innerProducts}.

\subsection{Monotonocity of the ternary Shannon  entropy}

The following result will be a useful tool in the course of our analysis.
\begin{thm}[Descartes' rule of signs~\cite{meserve1982fundamental,henrici1974applied}]\label{thm:Descartes}
  Let $p(x)=\sum_{m=0}^{s}a_{n-m}x^{n-m}$ be a real polynomial of order $n$ where $s\leq n$ and $a_{n-m}\neq0$. Then the number of positive real zeros (including multiplicities) is equal to $V-2k$ where $k\geq0$ and $V$ is the number of sign variations of $a_{n-m}$ starting from $a_n$.
\end{thm}
\begin{lem}\label{lem:pqProps}
  Let $\ve\leq0$. Then $q(z)$ in~\eqref{eq:pFcn} is  monotone-decreasing and concave in $z\in(0,1)$ for all $p_k$. It has a single positive root $z^\#\in(0,1)$ iff $\g>0$ in which case $q(z)\geq0$ for $z\in(0,z^\#)$.
\end{lem}
\begin{proof}
The monotonicity of $q$ follows from
  $$
  q'=2\d z+3\ve z^2,
  $$
  since $\d,\ve\leq0$.  Because of Theorem~\ref{thm:Descartes} (or just by inspection), there is no positive root of $q(z)$ for $\g<0$ again following from $\d,\ve\leq0$. There is one positive root for $\g>0$ and it has to lie in the interval $(0,1]$ since $q(0)=\g>0$ and
  $$
  q(1)=\g+\d+\ve\leq\g+\d=-{2\over27}+2p_0p_1p_2\leq0
  $$
  valid  for all $p_k$.
\end{proof}
\begin{rem}
  Even more straightforward is to show $p<0$  in $z\in(0,1)$ (follows from Eqs.~\eqref{eq:pFcn},~\eqref{eq:pqAlpha} and~\eqref{eq:pqBeta} by considering $(p_0+p_1+p_2)^2=1$). The equality $p=0$ is achieved for $z=0$ and $p_0=p_1=p_2=1/3$ but in order to have future expressions well-defined we will consider the open interval $z\in(0,1)$ throughout this work.  Similarly, we find $p'\leq0$.
\end{rem}
It is useful to know the generic behavior of the central piece of the cubic solutions, Eq.~\eqref{eq:tkTrig}. That is uncovered in the following lemma.
\begin{lem}\label{lem:gDifBehavior}
  Let $\ve\leq0$ and
  \begin{equation}\label{eq:core}
    g(z)={3\over2}{q\over p}\sqrt{-{3\over p}}.
  \end{equation}
  Then $|g(z)|\leq1$, $g(z)\propto -q(z)$ and $g'\lessgtr0$ for $z\in(0,1)$.
\end{lem}
\begin{proof}
  The bound $|g(z)|\leq1$ follows from the cubic equation discriminant
   \begin{equation}\label{eq:discriminant}
            {q^2\over4}+{p^3\over27}\leq0,
    \end{equation}
  where the inequality is always true for the case of three real roots of a cubic equation~\cite{waerden1966vol}. This, on the other hand, must be true since $\varpi$ is a density matrix. Eq.~\eqref{eq:core} can be both positive and negative with its sign always opposite to that of $q(z)$. This is because ${1\over p}\sqrt{-{3\over p}}<0$ for $z\in(0,1)$ following from Lemma~\ref{lem:pqProps}. A~related useful fact is that for $\g<0$ we get $g(z)>0$ for $z\in(0,1)$. Finally, by writing
  \begin{equation}\label{eq:gDif}
    g'={3\sqrt{3}\over4}{2pq'-3p'q\over p^2\sqrt{-p}}
  \end{equation}
  and noticing that the denominator is nonnegative  we only need to study the behavior of $\nu_1(z)=2pq'-3p'q$. First, we find a zero root due to $\nu_1(z)=-2z\left(-2\a\d+3\b\g-3z\a\ve +z^2\b\d\right)$. The quadratic equation $-2\a\d+3\b\g-3z\a\ve +z^2\b\d=0$  yields two other real roots and, in general, they both may lie in the interval $(0,1)$. Only when  $\g\geq0$, one of the roots is negative.
\end{proof}
\begin{lem}\label{lem:taunDif}
  Let $\tau(z,n)=\sqrt{-p}\cos{h\over n}$ and $n\in\bbZ_{>1}$ such that $p,p'<0$, $0\leq h\leq\pi$ in $z\in(0,1)$ and $h'>0$ in $\euI\subset(0,1)$. Then ${\dif{\tau(z,n)}\over\dif{z}}>{\dif{\tau(z,2)}\over\dif{z}}$ in $\euI$.
\end{lem}
\begin{proof}
  We find
  \begin{equation}\label{eq:taunDif}
    {\dif{\tau(z,n)}\over\dif{z}}=\frac{2ph'\sin{h\over n}-np'\cos{h\over n}}{2n\sqrt{-p}}.
  \end{equation}
  The denominator is positive for $z\in(0,1)$ but there are two competing expressions in the numerator. The first summand is negative, the second one is nonnegative and so the overall sign may be hard to infer.  The inequality follows by observing that the nonnegative summand in the numerator of~\eqref{eq:taunDif} remains constant as $n$ increases while the negative one is divided by $n$ and so its overall contribution diminishes. Finally, $\sin{h\over n}$ and $\cos{h\over n}$ do not change their sign with a growing $n\geq2$ and, conveniently, $\sin{h\over n}>\sin{h\over n+1}$ holds together with $\cos{h\over n}<\cos{h\over n+1}$ for $n\geq2$ as illustrated in Fig.~\ref{fig:acosaasin}.
\end{proof}
  \begin{figure}[h]
   \resizebox{13cm}{!}{\includegraphics{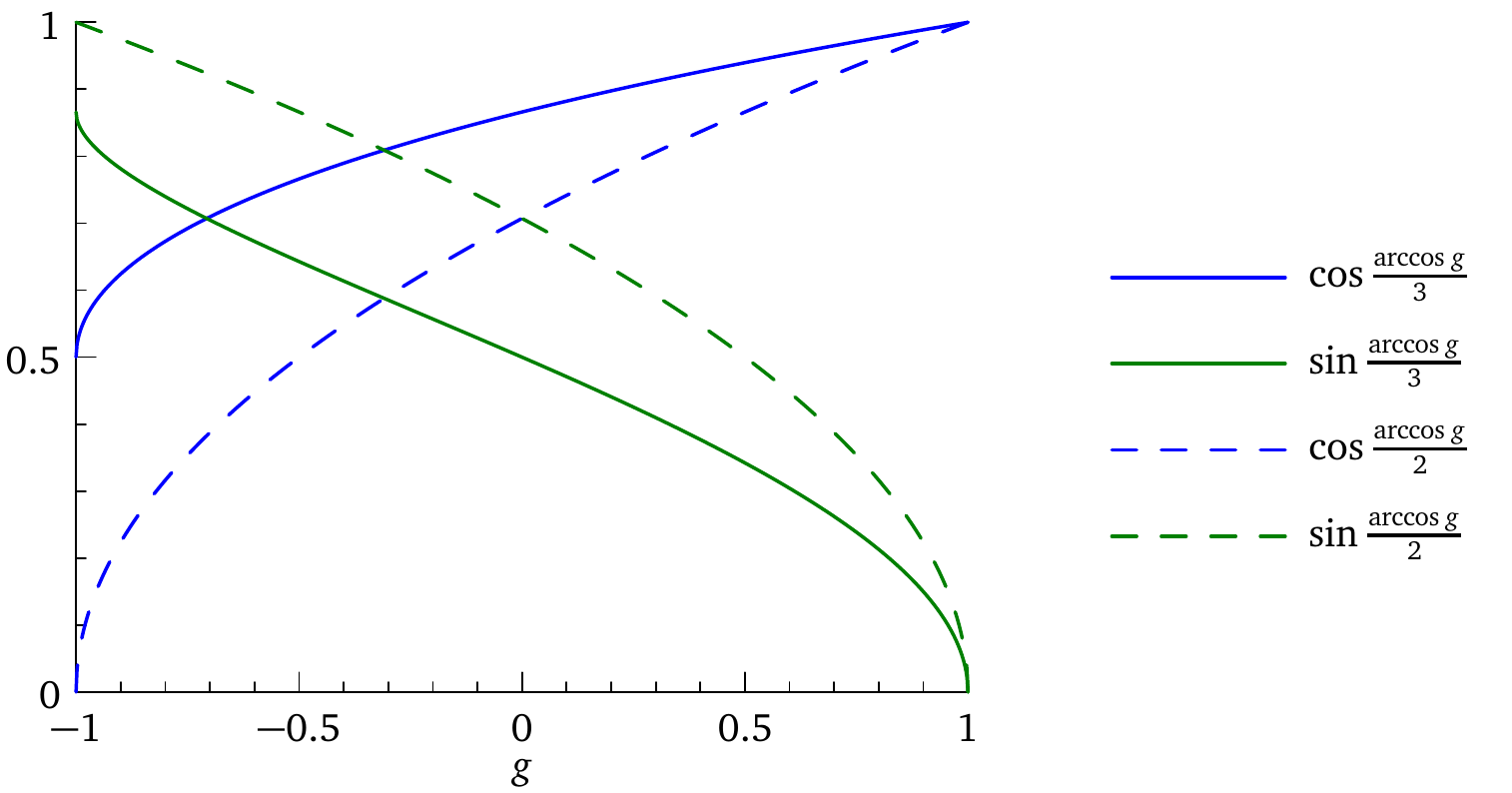}}
    \caption{Properties of some trigonometric functions.}
    \label{fig:acosaasin}
\end{figure}
\begin{prop}\label{prop:t0Dif}
  The function $t_0$ is monotone-increasing in $z\in(0,1)$ for $\ve\leq0$.
\end{prop}
\begin{proof}
  From~\eqref{eq:tkTrig} we get
  \begin{equation}\label{eq:t0Dif}
    t'_0={2\over3}\frac{\sqrt{-p(z)} g'(z) \sin{\left(\frac{1}{3} \arccos{(g(z))}\right)}}{ \sqrt{1-g(z)^2}}
    -\frac{p'(z) \cos{\left(\frac{1}{3} \arccos{(g(z))}\right)}}{\sqrt{-p(z)}}.
  \end{equation}
  Both denominators are non-negative (check zero). Since $|g(z)|\leq1$ and $\range{[\arccos]}=[0,\pi]$, both trigonometric functions are nonnegative. So the second summand (without the minus) is always negative due to $p'\leq0$. The overall expression is thus positive whenever $g'\geq0$. But from Lemma~\ref{lem:gDifBehavior} we know that $g'<0$ can occur as well so let us assume that for the rest of the proof. For $h\df\arccos{g}$  we first observe
    \begin{equation}\label{eq:hDif}
    h'=-{g'\over\sqrt{1-g^2}}
  \end{equation}
  and so $h'>0$. We now summon Lemma~\ref{lem:taunDif} and for it to be useful we show  ${\dif{\tau(z,2)}\over\dif{z}}\geq0$. The case $n=2$ is special since we can use the half-angle formula $\cos{h\over2}=\sqrt{1+\cos{h}\over2}$ (valid for $-\pi\leq h\leq\pi$)  to be inserted into $\tau(z,2)=\sqrt{-p}\cos{h\over 2}$ and we get
  \begin{equation}\label{eq:tau2Dif}
    {\dif{\tau(z,2)}\over\dif{z}}={ph'\sin{h}-p'(1+\cos{h})\over2\sqrt{-2p(1+\cos{h})}}
    ={-pg'-p'(1+g)\over2\sqrt{-2p(1+g)}}=-{\left(p(1+g)\right)'\over2\sqrt{-2p(1+g)}}\geq0.
  \end{equation}
  The inequality $(p(1+g))'\leq0$  follows from $p(1+g)$ being monotone-decreasing since both $p$ and $pg\approx q\sqrt{-3/p}$ are  monotone-decreasing (from Lemma~\ref{lem:pqProps} we know that for $\ve\leq0$ the function $q\lessgtr0$ is monotone-decreasing and $\sqrt{-3/p}$  as well and because $\min_{p_k,z}\limits{[\sqrt{-3/p}]}=3$ then $\sqrt{-3/p}$ merely ``stretches''~$q$). A sum of decreasing functions is decreasing   which concludes the proof since according to Lemma~\ref{lem:taunDif} we have
  $$
  t_0'{\sqrt{3}\over2}\equiv{\dif{\tau(z,3)}\over\dif{z}}>{\dif{\tau(z,2)}\over\dif{z}}\geq0.
  $$
%
%
\end{proof}
\begin{prop}\label{prop:t2Dif}
    The function $t_2$ is monotone-decreasing in $z\in(0,1)$  for $\ve\leq0$.
\end{prop}
\begin{proof}
  Considering $t_k$ in~\eqref{eq:tkTrig} as functions of $p$ and $q$, it is known~\cite{waerden1966vol} that $t_2(p,q)=-t_0(p,-q)$. The mapping $q\mapsto-q$ changes the sign of $g$ (and so of $g'$ as well, see~\eqref{eq:core} and~\eqref{eq:gDif}). The proof of Proposition~\ref{prop:t0Dif} then goes through in the same way as for $t_0(p,-q)$ since the trigonometric functions in~\eqref{eq:t0Dif} remain nonnegative for $-g$ and the proof ``covers'' both cases $g'\geq0$ and $g'<0$. Hence, $t'_0(p,-q)>0$  and we conclude that $t'_2<0$.
\end{proof}
\begin{cor}\label{cor:t0plust1incrs}
  The function $t_0+t_1$ is monotone-increasing in $z\in(0,1)$  for $\ve\leq0$.
\end{cor}
\begin{proof}
  From~Eq.~\eqref{eq:tkPropsSumZero} we find $t_0+t_1=-t_2$ and because $t_2$ is monotone-decreasing, its negative is monotone-increasing.
\end{proof}
\begin{rem}
  Notice that we do not claim anything about the monotonocity of $t_1$ and indeed it in general does not hold.
\end{rem}
\begin{defi}[\cite{roberts1974convex}]
  A function is called Schur-concave iff it is concave and symmetric.
\end{defi}
\begin{defi}
  Let $\vec{u}$ be an $\ell$-tuple for a non-increasingly ordered sequence $u_0\geq\hdots\geq u_{\ell-1}$ denoted as $u_k^\downarrow$ where $u_k\geq0$. We say that $\vec{u}$ is majorized by $\vec{v}$ (written as $\vec{u}\prec\vec{v}$) iff
  \begin{align}
    \sum_{k=0}^{m-1} u_k^\downarrow &\leq  \sum_{k=0}^{m-1} v_k^\downarrow, \label{eq:majorIneq}\\
    \sum_{k=0}^{\ell-1} u_k & = \sum_{k=0}^{\ell-1} v_k\label{eq:majorEq}
  \end{align}
  is satisfied for $0\leq m\leq\ell-1$.
\end{defi}
\begin{thm}[\cite{roberts1974convex}, Karamata~\cite{karamata1932inegalite}]\label{thm:SchurConc}
  If a function $f(u)$ is concave then $ f(\vec{u})\df\sum_k f(u_k)$ is Schur-concave and
  \begin{equation}\label{eq:reverseOrder}
    \vec{u}\prec\vec{v}\Rightarrow f(\vec{u})\geq f(\vec{v}).
  \end{equation}
\end{thm}
\begin{rem}
The function $f(u)=-u\log{u}$  is concave in $u\in(0,1)$ and therefore the Shannon entropy  $S(\vec{u})\df-\sum_{k=0}^Ku_k\log{u_k}$ is a Schur-concave function. For $K=2$, we obtain the ternary Shannon  entropy $h_3(\vec{u})$, Eq.~\eqref{eq:ternaryEnt}.
\end{rem}
\begin{thm}\label{thm:entrDecreasing}
  The von Neumann entropy $H(\varpi)$ of density matrix~\eqref{eq:Adm} is a monotone-decreasing function of the overlap $z$ as introduced in~\eqref{eq:cdSpecCase}, for all $p_k$ and for all $0\leq\vt\leq\pi/2$ corresponding to $\ve\leq0$ in~\eqref{eq:pqEpsi}.
\end{thm}
\begin{proof}
  The eigenvalues of $\varpi$ are $x_0\geq x_1\geq x_2$ and satisfy $\sum_kx_k=1$. From the discussion below Eq.~\eqref{eq:tkPropsSumZero} we know when $x_2\geq0$ holds and so $(x_k)_{k=0,1,2}$ is a probability distribution. We will identify $\vec{u}=\vec{x}(z_1)$ and $\vec{v}=\vec{x}(z_2)$ where $0< z_1\leq z_2<1$. Then, Proposition~\ref{prop:t0Dif} and Corollary~\ref{cor:t0plust1incrs} imply~\eqref{eq:majorIneq}. Since  $\sum_{k=0}^2u_k=\sum_{k=0}^2v_k=1$ holds (so~\eqref{eq:majorEq} is satisfied) we may write $\vec{u}\prec\vec{v}$. Following Theorem~\ref{thm:SchurConc} we obtain $h_3(\vec{u})\geq h_3(\vec{v})$ and $H(\varpi)\equiv h_3$ concludes the proof.
\end{proof}

\subsection{Concavity of the ternary Shannon  entropy}
We now turn our attention to the concavity proof. We start by a proving the convexity of $t_0$ in~\eqref{eq:tkTrig} in $z$. To that end, we set $\tau(z,n)=\sqrt{-p}\cos{h\over n}$ as in Lemma~\ref{lem:taunDif} and study the properties of its second derivative
\begin{align}\label{eq:taunDif2}
  {\diff{\tau(z,n)}\over\dif{z^2}}&={1\over4(-p)^{3/2}}
  \bigg(
  \cos{h\over n}\,{-4p^2h'^2+n^2(2pp''-p'^2)\over n^2}+\sin{h\over n}\,{-4p(ph')'\over n}
  \bigg).
\end{align}
To show ${\diff{\tau(z,3)}\over\dif{z^2}}\geq0$ we have to separately investigate several different cases. We will need a couple of auxiliary results.

The proof of concavity will be presented for $\ve\leq0$. The next lemma is the only exception where $\ve$ is arbitrary.
\begin{lem}\label{lem:posPart}
  We find
  $$
  {-4p^2h'^2+n^2(2pp''-p'^2)\over n^2}\geq0
  $$
  for $p$ in~\eqref{eq:pFcn}, $h=\arccos{g}$ for $g$ given by~\eqref{eq:core} and $n=3$.
\end{lem}
\begin{proof}
  Using~\eqref{eq:core} in~\eqref{eq:hDif} the inequality becomes
  \begin{align}
     &27(3qp'-2pq')^2-9(4p^3+27q^2)(2pp''-p'^2)=(\a+\b z^2)\nu_4(z)\leq0,
  \end{align}
  where
  \begin{equation}\label{eq:nu4}
    \nu_4(z)=\b(4 \a^3 +27\g^2)+z^2(8 \a^2 \b^2+12\a\d^2+18\b\g\d)+36z^3\a\d\ve+z^4\big(\a(4 \b^3+27 \ve^2)+3\b\d^2\big).
  \end{equation}
  Since $\a+\b z^2\equiv p\leq0$ we have to show that $\nu_4(z)$ is nonnegative for $z\in(0,1)$. Theorem~\ref{thm:Descartes} reveals a lot of information about $\nu_4$ through its coefficients. Depending on the sign of $\ve$ we find from~\eqref{eq:pqCoeffs}
    \begin{center}
       \begin{tabular}{@{} *5cl @{}}    \toprule
        \emph{monomial degree} & \emph{sign for $\ve<0$} & \emph{sign for $\ve>0$} \\\midrule
         4  & $+$  &  $+$  \\
         3  &  $-$   &  $+$ \\
         2  & $\pm$  & $\pm$\\
         0  & $+$  & $+$ \\ \bottomrule
         \hline
        \end{tabular}
    \end{center}
  No matter what the signs of $8 \a^2 \b^2+12\a\d^2+18\b\g\d$ and $\ve$ are there are always two sign changes. Hence $\nu_4(z)$ has two or none positive roots (for $\ve=0$ the second row from the top is missing but still there can be two or none positive roots). Let's first assume $\ve=-2p_0p_1p_2$ (its minimal value given by $\vt=0$). In this case we observe
  $$
  \nu_4(1)=4\a^3\b+8\a^2\b^2+\a\big(4\b^3+3 (2\d +3 \ve)^2\big)+3\b(3\g+\delta )^2=0.
  $$
  Since $\nu_4'(1)=0$ as well and $\nu_4''(1)=4\left(4\a^2\b^2+3\a(4 \b^3+2 \d^2+18\d\ve +27 \ve^2)+9\b\d(\g+\d)\right)\geq0$ the point $z=1$ is a proper local minimum. The expression $\nu'_4$ is a cubic polynomial. Hence it has three roots: one of them is always zero and the greatest one always equals one (the one we found previously). The third root can be both positive or negative and whatever its position is we want to make sure that $\nu_4\geq0$ in the interval $(0,1)$. Recall that according Theorem~\ref{thm:Descartes} there must be another positive root of $\nu_4$. At first sight it seems impossible because if the third root of $\nu_4'$ is negative then the segment of $\nu_4$ in $(0,1)$ must be decreasing ($\nu_4(0)=\b(4 \a^3 +27\g^2)\geq0$). Even if the third root of $\nu_4'$ lies in $(0,1)$ it can be either a~\emph{positive} local maximum or a stationary point. This is because we showed that $z=1$ is a local minimum, $\nu_4'$ has only three roots and again because of $\nu_4(0)\geq0$. So where is the remaining positive root? The only possibility is that $z=1$ is a double root. Indeed, by calculating the discriminant~\cite{waerden1966vol} of $\nu_4$ we find it to be equal to zero. This means that at least two roots coincide. Hence $\nu_4\geq0$ holds for $\ve=-2p_0p_1p_2$.

  For $0<\vt\leq\pi/2$ the coefficients of the monomials of order 3 and 4 in~\eqref{eq:nu4} clearly increase and hence no new root can appear in the interval $(0,1)$. For $\pi/2<\vt\leq\pi$, the monomial order 4 coefficient decreases but $\ve^2\propto\cos^2{\vt}$ is a symmetric function and we have seen that $\nu_4(z)$ had no positive root even when $\ve<0$. But now $\ve>0$ and so again there is no positive root which concludes the proof.
\end{proof}
\begin{rem}
  The claim holds for any $n\geq3$ but we do not make use of it.
\end{rem}

\begin{lem}\label{lem:pdgDiffBehavior}
  The function $(pg')'(z)$ has a single positive root $z^*$  whenever $\g(z)\geq0$ and  $(pg')'(z)\leq0$ for $z\in(0,z^*)$.
\end{lem}
\begin{proof}
  We calculate
  \begin{equation}\label{eq:pDgDif}
    (pg')'={3\sqrt{3}\over8}{q(9p'^2-6pp'')+4p(-2p'q'+pq'')\over\sqrt{-p}p^2}.
  \end{equation}
  The position of the positive roots is unaffected by the numerical prefactors or by the denominator. Hence, we rewrite only the numerator $\nu_2(z)=q(9p'^2-6pp'')+4 p (-2p'q'+pq'')$ in terms of Eqs.~\eqref{eq:pFcn} and~\eqref{eq:qFcn}:
  \begin{equation}
    \nu_2(z)=-12\alpha\beta\ve z^3+\left(24 \beta ^2 \gamma -28 \alpha  \beta  \d\right)z^2+24\a^2\ve z+8\a^2\d-12\a\b\g.
  \end{equation}
  Given $\a,\b,\d,\ve\leq0$ and $\g>0$ there is only one sign change if $8\a^2\d-12\a\b\g\leq0$. This is indeed satisfied for $\g>0$. The observation
  $$
   \lim_{z\to+\infty}{[(pg')']}=+\infty
  $$
  concludes the proof.
\end{proof}
\begin{rem}
  In fact, we can refine the previous lemma by calculating
  $$
  \lim_{z\to0}{[(pg')']}=-\frac{3}{2} \sqrt{3} \left(-\frac{1}{\alpha }\right)^{3/2} (2 \alpha  \delta -3 \beta  \gamma )\leq0
  $$
  and expressing $\nu_2(z)$ with the help of~\eqref{eq:pqCoeffs} and $\sum_{i}p_i=1$ as
    \begin{equation}
      \nu_2(1)={8\over9}\big(p_0^4+2 p_0^3 (-1+p_1)+(-1+p_1)^2 p_1^2+p_0p_1 (-1-p_1+2 p_1^2)+p_0^2 (1-p_1+3 p_1^2)\big)\geq0.
    \end{equation}
  Therefore, $z^*\in(0,1)$. The minimum on the RHS is achieved for $p_0=p_1=p_2=1/3$.
\end{rem}
\begin{lem}\label{lem:pdhDiffBehavior}
  The function $(ph')'(z)$ has a single positive root for $z\in(0,1)$ whenever $\g\geq0$ and for all $\ve\leq0$.
\end{lem}
\begin{proof}
  We write
  \begin{equation}\label{eq:DpDh}
    (ph')'=-\frac{gpg'+(1-g^2)(pg')'}{(1-g^2)^{3/2}}
  \end{equation}
  and after inserting Eqs.~\eqref{eq:gDif},~\eqref{eq:core} and
  \begin{equation}\label{eq:gDif2}
    g''(z)=\frac{3\sqrt{3}\left(4p(p q''-3p'q')+3q(5p'^2-2pp'')\right)}{8(-p)^{-3/2}p^5}
  \end{equation}
  we get an expression whose numerator reads
  \begin{equation}\label{eq:nu5}
    \nu_5=-81q^3 p''+54q^2\left(p'q'+pq''\right)+q\left(-12p^3p''+18p^2p'^2-54pq'^2\right)+8p^4q''-16p^3p'q'
  \end{equation}
  and whose denominator is negative in $(0,1)$. By inserting~Eqs.~\eqref{eq:pq} we get a daunting polynomial of degree seven:
  \begin{align}\label{eq:nu5deg7}
        \nu_5&=z^7 \left(-24\a\b^3\ve-162\a\ve^3-54\b\d^2\ve \right)+z^6\left(-56\a\b^3\d -378\a\d\ve^2+48\b^4\g +324\b\g\ve^2-54\b\d^3\right)\nn\\
        &\quad+z^5 \left(324 \b  \g  \d  \ve -324 \a \d^2 \ve \right)+z^4 \left(-96 \a^2 \b^2 \d +72 \a \b^3 \g +162 \a \g  \ve^2-108 \a \d^3-54 \b\g\d^2\right)\nn\\
        &\quad+z^3 \left(72 \a^3 \b  \ve +216 \a \g  \d  \ve +162 \b  \g^2 \ve \right)+z^2 \left(-24 \a^3 \b  \d -162 \b  \g^2 \d \right)+
        z\left(48 \a^4 \ve +324 \a \g^2 \ve \right)\nn\\
        &\quad+16\a^4\d-24\a^3\b\g+108\a\g^2\d-162\b\g^3.
  \end{align}
  Let us first assume $\ve=-2p_0p_1p_2$ which is the minimal value given by $\vt=0$. Then, there is an inflection point at $z=1$: $\nu_5'(1)=\nu_5''(1)=0$. This indicates a triple root (corroborated by the zero discriminant indicating multiple roots) and so
  \begin{equation}
       \nu_5=f_4(z-1)^3,
  \end{equation}
  where $f_4=\sum_{i=0}^4a_iz^i$. By comparing the coefficients with~\eqref{eq:nu5deg7} we deduce the coefficient $a_i$ and get
  \begin{align}\label{eq:f4}
          f_4&= -6z^4 \ve \left(4 \a \b^3+27 \a \ve^2+9 \b \d ^2\right)\nn\\
          &\quad+2z^3 \left(-4 \a \b^3 (7 \d +9 \ve)-27 \a \ve^2 (7 \d +9 \ve)+24 \b^4 \g-27\b(-6\g\ve^2+\d^3+3 \d^2 \ve)\right)\nn\\
          &\quad-6 z^2 (4 \a^3+27 \g^2) (\a (4 \d +6 \ve)-\b (6 \g+\d ))\nn\\
          &\quad-6 z (4 \a^3+27 \g^2) (2 \a (\d +\ve)-3 \b \g)\nn\\
          &\quad-2(4\a^3+27\g^2)(2\a\d-3\b\g).
  \end{align}
  With the help of the following table
    \begin{center}
       \begin{tabular}{@{} *5cl @{}}    \toprule
        \emph{monomial degree} & \emph{sign}  \\\midrule
         4  & $-$      \\
         3  &  ? \\
         2  & $+$   \\
         1  & $+$   \\
         0  & $+$   \\\bottomrule
         \hline
        \end{tabular}
        \label{tab:nu5}
    \end{center}
  Theorem~\ref{thm:Descartes} reveals that there is only one positive root. Note that the degree three coefficient of~\eqref{eq:f4} seems too complicated to analytically deduce its sign but our ignorance does not affect the number of sign variations. Now we show that by for any $\ve\leq0$ the single root shifts and the inflection disappears. We inspect the coefficients of~\eqref{eq:nu5deg7} where $\ve$ appears. Considering $\g\geq0$, the ones accompanying the monomials $z,z^3$ and $z^5$ satisfy
  \begin{subequations}
    \begin{align}
      48 \a^4 \ve +324 \a \g^2 \ve & \leq 0, \\
      72 \a^3 \b  \ve +216 \a \g  \d  \ve +162 \b  \g^2 \ve & \leq 0,\\
      324 \b  \g  \d  \ve -324 \a \d^2 \ve & \leq 0.
    \end{align}
  \end{subequations}
  Hence an increase of $\ve$ from its minimal values to any $\ve\leq0$ will not add a new root in $(0,1)$.  Similarly for the $z^7$ coefficient $\ve(-24\a\b^3-162\a\ve^2-54\b\d^2)$ which, due to
  \begin{equation}\label{eq:nu5z7coeff}
    -24\a\b^3-162\a\ve^2-54\b\d^2\geq0
  \end{equation}
  (valid only for the minimal $\ve$), is an increasing function of $\ve\leq0$. This is because $\a\leq0$ and so~\eqref{eq:nu5z7coeff} is a decreasing function of $\ve\leq0$ (\eqref{eq:nu5z7coeff} can become negative). Even if~\eqref{eq:nu5z7coeff} does not change the sign, the $z^7$ coefficient  will always be greater than the one with the minimal $\ve$ because  the overall multiplication by $\ve\leq0$ swaps the sign (and so the order). Finally, the $z^4$ and $z^6$ coefficients contain negative factors accompanying $\ve^2$ (recall $\a,\b,\d\leq0$ and $\g\geq0$ by assumption). Hence, as $\ve^2$ decreases, it effectively increases the coefficients of $z^4$ and $z^6$. We can conclude that no new root for $z\in(0,1)$ appears for $\ve\leq0$.
\end{proof}
\begin{prop}\label{prop:chainImpl}
  The following relations hold:

    \centering
      \begin{tikzcd}[row sep=large, column sep=normal]
       q\geq0 \arrow[r,Leftrightarrow,"\mathrm{Lemma~\ref{lem:gDifBehavior}}" inner sep=1ex] &   g\leq0 \arrow[r,Rightarrow,"\mathrm{(i)}"]\arrow[d,Rightarrow,"\mathrm{(ii)}"] & (pg')'\leq0 \\
         & (ph')'\geq0 &
    \end{tikzcd}
\end{prop}
\begin{proof}\mbox{}\\
  (i) The sought after implication can be reformulated in the language of Lemma~\ref{lem:pqProps} and~\ref{lem:pdgDiffBehavior} as $z^\#\leq z^*$ since $q\geq0$ in $(0,z^\#)$ and $(pg')'\leq0$ in $(0,z^*)$. We proceed by setting $q=0$ and, conveniently, the numerator of~\eqref{eq:pDgDif} simplifies to
  \begin{equation}\label{eq:nu2}
    \nu_2(z)\big|_{q=0}\propto-2p'q'+pq''= 6 \b\ve  z^3  +6 \beta  \delta  z^2-6 \alpha\ve z-2\alpha\delta.
  \end{equation}
  We ignored the factor $4p$ as it does not affect the position of the roots for $q=0$. In principle we just need to compare the position of the roots for the polynomials $q$ and $ \nu_2(z)|_{q=0}$. However, they are both cubic polynomials and the roots' form is too complicated to determine their  relation. It follows from  Lemma~\ref{lem:pqProps}, Lemma~\ref{lem:pdgDiffBehavior} and the previous remark that the polynomials intersect at a single point in the interval $(z^\#,z^*)\subset(0,1)$ and, in addition, the position of the intersection point above or below the $x$ axis informs us about the relation of the two roots. It would not be very helpful to  set $q=\nu_2(z)|_{q=0}$ and solve for $z$, though. It again leads to a cubic equation and we face a similar problem as before. The trick we will use is the following transformation:
  \begin{equation}\label{eq:qtilde}
    q(z)\mapsto \tilde{q}(z)=-6\b q=-6\b\g-6\b\d z^2-6\b\ve z^3.
  \end{equation}
  The new function $\tilde{q}(z)$ has the same properties as $q(z)$  uncovered in Lemma~\ref{lem:gDifBehavior} (the minus sign reverses the negative sign of $\b$). By setting  $\tilde{q}=\nu_2(z)|_{q=0}$ we obtain another cubic equation
  \begin{equation}
    \tilde{\mu}(z)=2\mu(z)=2\big(6\b\ve  z^3 +6\beta  \delta  z^2-3\alpha\ve z-\a\d+3\b\g\big)=0.
  \end{equation}
  Its (single) root in $(0,1)$ reveals where $\tilde{q}$ and $\nu_2(z)|_{q=0}$ intersect but that also means that by comparing the roots' position of $\mu$ (or $\tilde{\mu}$) with $\nu_2(z)|_{q=0}$ in the interval $(0,1)$ we learn whether $\tilde{q}$ and $\nu_2(z)|_{q=0}$ intersected above or below the $x$ axis. So by setting $\mu(z)=\nu_2(z)|_{q=0}$ we crucially get a linear equation whose solution reads
  \begin{equation}\label{eq:linRoot}
    z_\ell=\frac{-\a\d-3\b\g}{3\a\ve }.
  \end{equation}
  By inserting it back to $\nu_2(z)|_{q=0}$ we get
  \begin{equation}\label{eq:nu2zell}
    \nu_2(z_\ell)|_{q=0}
    =\frac{2\b\left(\a^3(27\gamma \ve^2+2 \d^3)+9 \a^2\b\g\d^2-27\b^3\g^3\right)}{9\a^3\ve^2}.
  \end{equation}
  It remains to show $\nu_2(z_\ell)|_{q=0}\leq0$ in order to prove  $z^\#\leq z^*$. Since $\a,\b\leq0$ it suffices to show that $\nu_3(z)=\a^3 \left(27\gamma \ve^2+2 \d^3\right)+9 \a^2\b\g\d^2-27\b^3\g^3\leq0$. Using~\eqref{eq:pqCoeffs} and $\sum_ip_i=1$ we find $\nu_3(z)\df{1\over27}f_1f_2$ where
  \begin{subequations}\label{eq:f1f2}
    \begin{align}
      f_1 & = \big(p_0+2 p_0^3+3 p_0^2 (-1+p_1)-3 p_0 p_1^2+p_1 (-1+3 p_1-2 p_1^2)\big)^2, \label{eq:f1}\\
      f_2 & = 4 p_0^6+12 p_0^5 (-1+p_1)+p_0^4 (13-27 p_1+24 p_1^2)+p_0^3 (-6+20 p_1-42 p_1^2+28 p_1^3)\nn\\
          & \quad + p_0^2 (1-5 p_1+24 p_1^2-42 p_1^3+24 p_1^4)+p_1^2 (1-3 p_1+2 p_1^2)^2+p_0 p_1^2 (-5+20 p_1-27 p_1^2+12 p_1^3).
    \end{align}
  \end{subequations}
  Since $f_1\geq0$, we have to show $f_2\leq0$. We reduced the problem to a task analytically solvable by Mathematica. Indeed, we find $\max{[f_2]}=0$ subject to $\g\geq0$ and $0\leq z_\ell\leq1$. The first inequality is a necessary condition for the initial assumption $q\geq0$ in $z\in(0,z^\#)$ via Lemma~\ref{lem:pqProps}. Achieving the maximum implies $\nu_3=0$ which in turn implies $ \nu_2(z_\ell)|_{q=0}=0$ (from~\eqref{eq:nu2zell}) and so $\mu(z_\ell)=0$ (see above~\eqref{eq:linRoot}). This finally leads to $\nu_2(z_\ell)|_{q=0}=\tilde{q}(z_\ell)=0=q(z_\ell)$ and so $z_\ell^\#=z_\ell^*$ which concludes the proof.\\
  (ii) Assuming $\g\geq0$ as a necessary condition to the current case of interest $q\geq0$ $(g\leq0)$  for $z\in(0,z^\#)$ (see Lemma~\ref{lem:pqProps} and~\ref{lem:gDifBehavior}) we find $\nu_5(0)=2(4\a^3+27\g^2)(2\a\d-3\b\g)\leq0$ (the different sign in the bottom of the table on page~\pageref{tab:nu5} is due to $f_4$ being multiplied by $(z-1)^3$) and so $(ph')'(0)\geq0$. This is because $(ph')'\propto-\nu_5$. We also notice that $(ph')'\geq0$ for $g=0$. This follows from Eq.~\eqref{eq:DpDh} implying that in this case $(ph')'=-(pg')'$. But from item (i) of the current lemma we know that $(pg')'\leq0$ for $g\leq0$. Inevitably, the only positive root of $(ph')'$ occurs for $g\geq0$, that is, as long as $g\leq0$ we get $(ph')'\geq0$ as we wanted to show.
\end{proof}
\begin{rem}
  Note that $(ph')'(0)\geq0$ does not contradict $g'(0)=0$ we found in Lemma~\ref{lem:gDifBehavior}. This could be hastily concluded by looking at Eq.~\eqref{eq:hDif}. But it is true only if $g'=0$ and $\sqrt{1-g^2}\neq0$. In many cases it is found, however, that for $z=0$ one gets $g'=\sqrt{1-g^2}=0$ and $\lim_{z\to0^+}h'\neq0$.
\end{rem}
\begin{lem}\label{lem:posPart2}
    Let $h=\arccos{g}$ and $(ph')'\leq0$. Then
    \begin{equation}\label{eq:lowerBound}
      \cos{h\over 2}\,{-4p^2h'^2+3^2(2pp''-p'^2)\over4\times3^2(-p)^{3/2}}+\sin{h\over 2}\,{-p(ph')'\over3(-p)^{3/2}}\geq0
    \end{equation}
    whenever $g\geq0$ and for all $\ve\leq0$.
\end{lem}
\begin{rem}
  The expression resembles part of~Eq.~\eqref{eq:taunDif2}. However, notice  $n=2$ in the trigonometric functions and $n=3$ elsewhere.
\end{rem}
\begin{proof}
  Given $\tau(z,n)=\sqrt{-p}\cos{\arccos{g}\over n}$, a straightforward calculation reveals
  \begin{align}\label{eq:preLowerBound}
    &{\diff{\tau(z,n)}\over\dif{z^2}} \nn\\
    &=\cos{\arccos{g}\over n}\,{4p^2g'^2+n^2(-1+g^2)(2pp''-p'^2)\over4n^2(1-g^2)\sqrt{-p}p}+
      \sin{\arccos{g}\over n}\,{-pgg'^2+(-1+g^2)(pg')'\over n\sqrt{1-g^2}(1-g^2)\sqrt{-p}}.
  \end{align}
  We set $n=2$ in the trigonometric functions and $n=3$ elsewhere at which point~\eqref{eq:preLowerBound} becomes the studied expression~(Eq.~\eqref{eq:lowerBound}) by virtue of~\eqref{eq:hDif}. We multiply both summands by $p\sqrt{-p}(1-g^2)\leq0$ and use $\cos{x\over2}=\sqrt{1+\cos{x}\over2}$ ($-\pi\leq x\leq\pi$) and $\sin{x\over2}=\sqrt{1-\cos{x}\over2}$ ($0\leq x\leq2\pi$). We got the reverse inequality to prove
  \begin{equation}\label{eq:kappa1}
    \kappa={4\over9}p^2g'^2(1-2g)+(-1+g^2)\Big((1+g)(2pp''-p'^2)+{4\over3}p(pg')'\Big)\leq0.
  \end{equation}
  For this purpose, we use~Eq.~\eqref{eq:core} and deduce
    \begin{equation}\label{eq:kappa2}
      \kappa={1\over p^2}\bigg({-}\nu_4+f_3{1\over2}\sqrt{3\over-p}\bigg),
    \end{equation}
   where $\nu_4$ is given by~\eqref{eq:nu4} and
  \begin{equation}\label{eq:f3}
    f_3(q)=-27 q^2q''+q(-6 pp'^2+18 q'^2)-4p^2(pq''-2p'q').
  \end{equation}
  Since in Lemma~\ref{lem:posPart} we proved $\nu_4>0$ we only have to show $f_3\leq0$ for $\kappa\leq0$ to hold. The inequality $f_3\leq0$ does not hold in general, however. That is not a problem as long as we show that it holds for $(ph')'\leq0$. First we assume $\g\geq0$. By contrapositive of Proposition~\ref{prop:chainImpl} (ii) we know
  \begin{equation}\label{eq:contrapositive}
    (ph')'\leq0\Rightarrow g\geq0\Leftrightarrow q\leq0.
  \end{equation}
  Hence we need to show $q\leq0\Rightarrow f_3\leq0$. For $q=0$ the function $f_3$ becomes $-4p^2(pq''-2p'q')$ which is proportional to $\nu_2(z_\ell)|_{q=0}$ (see~\eqref{eq:nu2}). Its relation to $q$ was studied in Proposition~\ref{prop:chainImpl} and we found $\nu_2(z_\ell)|_{q=0}=q(z_\ell)=0$ for $z_\ell$ given by~\eqref{eq:linRoot}. Therefore, $f_3(0)$=0. Then, as demanded in~\eqref{eq:contrapositive}, for any $q<0$ we get $f_3(q)<0$ since~Eq.~\eqref{eq:f3} is an increasing function of $q$. This follows from  $q''=6\ve z\leq0$ (valid for $\ve\leq0$) and $-6pp'^2+18q'^2\geq0$ by looking at~Eqs.~\eqref{eq:pq} and~\eqref{eq:pqCoeffs}. For $\g<0$ we know from Lemma~\ref{lem:gDifBehavior} that $g>0$ $(q<0)$ always holds independently on the sign of $(ph')'$. Therefore $f_3<0$ and the proof goes as outlined above.
\end{proof}
\begin{prop}\label{prop:t0Dif2}
  The function $t_0$ is convex in $z\in(0,1)$.
\end{prop}
\begin{proof}
  The function $t_0$ is proportional to $\tau(z,3)=\sqrt{-p}\cos{\arccos{g}\over3}$ and so we will focus on proving ${\diff{\tau(z,3)}\over\dif{z^2}}\geq0$ given by~\eqref{eq:taunDif2} for $n=3$. In Lemma~\ref{lem:posPart} we presented a proof of nonnegativity of a fraction multiplying $\cos{\arccos{g}\over3}$. Both  $\cos{\arccos{g}\over3}\geq0$ and $\sin{\arccos{g}\over3}\geq0$ for $|g(z)|\leq1$ and so ${\diff{\tau(z,3)}\over\dif{z^2}}\geq0$ holds whenever $(ph')'\geq0$. For the rest of the proof assume $(ph')'<0$. We will construct a lower bound on  ${\diff{\tau(z,3)}\over\dif{z^2}}$ and show it to be nonnegative. To this end, we  summon the inequalities $\sin{h\over n}>\sin{h\over n+1}$ and $\cos{h\over n}<\cos{h\over n+1}$ (valid for $n\geq2$ and visible in Fig.~\ref{fig:acosaasin} for $n=2,3$) and substitute $\sin{h\over 3}$ and $\cos{h\over 3}$ by $\sin{h\over 2}$ and $\cos{h\over2}$, respectively. This is a lower bound on ${\diff{\tau(z,3)}\over\dif{z^2}}$ and the quantity was proved to be nonnegative in Lemma~\ref{lem:posPart2}. This concludes the proof.
\end{proof}
\begin{rem}
  The fact that $g\not<0$ for $(ph')'\leq0$ is crucial. First of all, it is not clear how to prove the validity of~${\diff{\tau(z,3)}\over\dif{z^2}}\geq0$ given by~\eqref{eq:taunDif2} for $(ph')'\leq0$. But even  the only manageable lower bound, Eq.~\eqref{eq:lowerBound}, is in some cases not good enough (i.e., non-negative) for $g<0$ and $(ph')'\leq0$.
\end{rem}
\begin{cor}\label{cor:t2concave}
  The function $t_2$ is concave in $z\in(0,1)$.
\end{cor}
\begin{proof}
  Similarly to Proposition~\ref{prop:t2Dif} we will make use of  $t_2(p,q)=-t_0(p,-q)$. The mapping $q\mapsto -q$ changes the sign of $g,g'$ and $h'$ (see Eqs.~\eqref{eq:core},~\eqref{eq:gDif} and~\eqref{eq:hDif}). Looking at~Eq.~\eqref{eq:taunDif2}, we notice that for $n\geq2$ the sign of the trigonometric functions remains unaffected (see Fig.~\ref{fig:acosaasin} for $n=2,3$). Similarly for the expression from Lemma~\ref{lem:posPart} coming from~\eqref{eq:taunDif2}. The sign change of $q$ also swaps~the sign of $(ph')'$ in~\eqref{eq:taunDif2}. This is because $(ph')'=p'h'+ph''$ together with
  $$
  h''=-{gg'^2+(1-g^2)g''\over(1-g^2)^{3/2}}
  $$
  taking into account that $g''$ changes the sign upon $q\mapsto -q$. But both cases ($(ph')'\geq0$ and $(ph')'<0$) have been separately investigated in Proposition~\ref{prop:t0Dif2}. So we conclude $t''_0(p,-q)\geq0$ and so $t''_2(p,q)\leq0$.
\end{proof}
\begin{lem}\label{lem:Hessian}
  Let $(u_i)_{i=0}^2$ be a probability distribution function and $h_3$ the ternary Shannon  entropy  defined in~\eqref{eq:ternaryEnt}.  Then $h_3$ is concave in $(0,1)\times(0,1)\subset\bbR^2$ and for a fixed $u_2\in(0,1)$ the function $h_3$ is monotone increasing (decreasing) for $u_1<(1-u_2)/2$ ($u_1>(1-u_2)/2$).
\end{lem}
\begin{proof}
  The Hessian matrix
  \begin{equation}\label{eq:Hessian}
    \mathsf{H}(h_3(\vec{u}))=\begin{bmatrix}
    -\frac{1}{1-u_1-u_2}-\frac{1}{u_1} & -\frac{1}{1-u_1-u_2} \\
    -\frac{1}{1-u_1-u_2} & -\frac{1}{1-u_1-u_2}-\frac{1}{u_2}   \\
     \end{bmatrix}
  \end{equation}
  is negative definite since $\Tr{}{\mathsf{H}}<0$ and $\det{\mathsf{H}}=1/(u_0u_1)+1/(u_0u_2)+1/(u_1u_2)>0$. The concavity of $h_3$ follows from the positivity of the characteristic polynomial throughout the interval $(u_0,u_1)\in(0,1)\times(0,1)$. We now fix the value of $u_2$ and the equation $\partial{h_3}/\partial{u_1}=0$ is satisfied for $u_1=(1-u_2)/2$. Thanks to the previously proved concavity, it is a local maximum for every $u_2\in(0,1)$ and thus $(1-u_2)/2$ defines a one-parameter family of local maxima for $h_3$.
\end{proof}
\begin{rem}
  Due to the symmetry between $u_2$ and $u_1$ in~\eqref{eq:ternaryEnt} we may fix $u_1$ and get a family of local maxima given by $-2u_2+1$. The global maximum of $h_3$ at $(u_2,u_1)=(1/3,1/3)$ lies in the intersection of $(1-u_2)/2$ and $-2u_2+1$. The situation is illustrated in Fig.~\ref{fig:h3}.
  \begin{figure}[h]
   \resizebox{10cm}{!}{\includegraphics{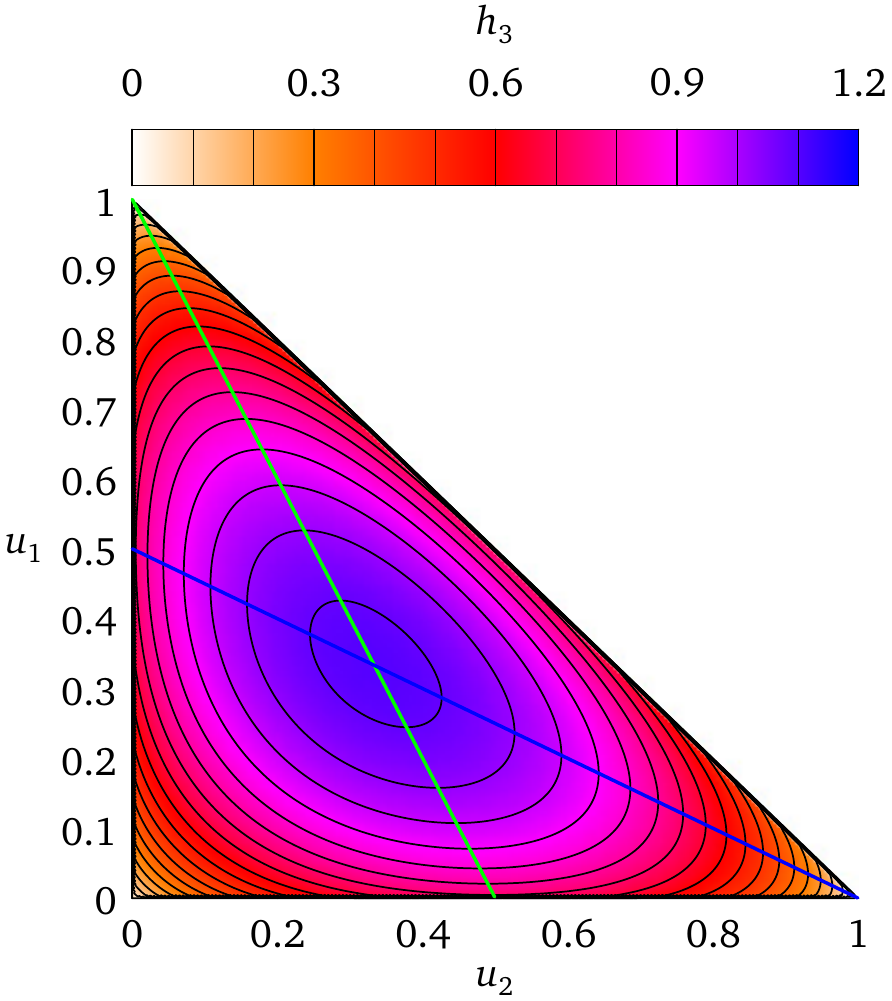}}
    \caption{The ternary Shannon  entropy Eq.~\eqref{eq:ternaryEnt} is shown. The blue line depicts $u_1=(1-u_2)/2$ while the green one is the plot of $u_1=-2u_2+1$.}
    \label{fig:h3}
\end{figure}
\end{rem}
\begin{thm}\label{thm:entrConcave}
  The von Neumann entropy $H(\varpi)$ of density matrix~\eqref{eq:Adm} is a concave function of the overlap $z$ as introduced in~\eqref{eq:cdSpecCase}, for all $p_k$ and for all $0\leq\vt\leq\pi/2$ corresponding to $\ve\leq0$ in~\eqref{eq:pqEpsi}.
\end{thm}
\begin{proof}
  The von Neumann entropy $H(\varpi)$ is given by $ h_3(\vec{x}(z))$ in~\eqref{eq:ternaryEnt} where $x_k=t_k+1/3$ come from Eq.~\eqref{eq:tkTrig}. Taking into  account $\sum_{i=0}^2x_i=1$ we get $\vec{x}:\bbR\mapsto\bbR^2$ and the investigated expression ${\mathrm{d^2}h_3\over\mathrm{d}z^2}$ will be written using the following notation: We define
    \begin{subequations}\label{eq:columnDerivs}
      \begin{align}\label{eq:columnXp}
          \mathsf{x'} & =\begin{bmatrix}
          x_1' \\
          x_2' \\
        \end{bmatrix} \\\label{eq:columnXpp}
          \mathsf{x''} & =\begin{bmatrix}
          x_1'' \\
          x_2'' \\
        \end{bmatrix}
        \end{align}
    \end{subequations}
    and
    \begin{equation}\label{eq:colh3p}
      \mathsf{h_3'}=\begin{bmatrix}
          \papa{h_3}{x_1} \\
          \papa{h_3}{x_2} \\
        \end{bmatrix}.
    \end{equation}
    Using the chain rule, the second derivative can be succinctly expressed as
    \begin{equation}\label{eq:h3Concavity}
      {\mathrm{d^2}h_3\over\mathrm{d}z^2}= (\mathsf{x'})^\top\mathsf{H}(h_3(\vec{x}))\,\mathsf{x'}+(\mathsf{h_3'})^\top\mathsf{x''},
    \end{equation}
    where $\top$ denotes transposition and the dot (matrix) product is implied. Hessian~Eq.~\eqref{eq:Hessian} is negative definite according to Lemma~\ref{lem:Hessian} and so the first summand is negative for any $\mathsf{x'}$. In order for the second summand to be nonpositive as well, one possibility is when either the functions $x_1$ and $x_2$ are concave and the two components of $h_3$ nondecreasing  or $x_1,x_2$ convex and $h_3$ entry-wise nonincreasing. We proved $x_2''\leq0$ in Corollary~\ref{cor:t2concave} but said nothing about the concavity of $x_1$. As a matter of fact, it is incomparably more difficult to prove $x_1''\leq0$ in spite of the overwhelming numerical evidence. The same numerics suggests that there is a whole class of input probabilities $p_k$ for which $x_1''=0$. So no `simple' bounds like those leading to Proposition~\ref{prop:t0Dif2} exist. But there is a third possibility of how to make the second summand in~\eqref{eq:h3Concavity} negative and it is the combination of the two previous cases. We know that $x_0''\geq0$ from Proposition~\ref{prop:t0Dif2} and $x_2''\leq0$ from Corollary~\ref{cor:t2concave}. The second summand~\eqref{eq:h3Concavity} will be negative if we take $x_0$ instead of $x_1$ in~Eqs.~\eqref{eq:columnDerivs} and~\eqref{eq:colh3p} and show $\papa{h_3}{x_0}\leq0$ and $\papa{h_3}{x_2}\geq0$. Note that the Hessian remains negative definite:
    \begin{equation}\label{eq:HessianAgain}
        \mathsf{H}(h_3(\vec{x}))=\begin{bmatrix}
        -\frac{1}{1-x_0-x_2}-\frac{1}{x_0} & -\frac{1}{1-x_0-x_2} \\
        -\frac{1}{1-x_0-x_2} & -\frac{1}{1-x_0-x_2}-\frac{1}{x_2}   \\
         \end{bmatrix}.
    \end{equation}
    Hence, in spite of the components of $\mathsf{x'}$ to have  different signs (see Propositions~\ref{prop:t0Dif} and~\ref{prop:t2Dif}), the summand is negative. Also note that we are proving the properties of the same ternary entropy~\eqref{eq:ternaryEnt} since it is equivalent to
    \begin{equation}\label{eq:h3Equiv}
        h_3(\vec{x}(z))=-x_0\log{x_0}-x_2\log{x_2}-(1-x_0-x_2)\log{[1-x_0-x_2]}.
    \end{equation}
    Lemma~\ref{lem:Hessian} informs us that $\papa{h_3}{x_2}\geq0$ and $x_2''\leq0$ together with $\papa{h_3}{x_0}\leq0$ and $x_0''\geq0$ is satisfied in the domain's subset delimited by the blue line ($x_0\geq(1-x_2)/2$) and the green line ($x_0\leq-2x_2+1$) depicted in Fig.~\ref{fig:h3} if instead of $u_2,u_1$ we have $x_2,x_0$ (resulting in the same figure). But it turns out that this is precisely the range of $\vec{x}$ represented by $x_0,x_2$. To this end, consider the basic property of the cubic roots~\cite{waerden1966vol} $t_0\geq t_1\geq t_2$ that becomes $x_0\geq x_1\geq x_2\geq0$ for the eigenvalues of~$\varpi$. First, using $x_0\geq x_1$ we write
    \begin{align}
      x_0 & \geq {x_0+x_1\over2}\nn\\
          & = {x_0+x_1+x_2-x_2\over2}\nn\\
          & = {1-x_2\over2},
    \end{align}
    where in the second row we used the normalization condition $\sum_ix_i=1$. The last equality leads to one of the desired bounds. For the second bound we start with $x_1\geq x_2$ to write
    \begin{align}
      x_0 & \leq x_0 -x_2+x_1\nn \\
          & = -2x_2+x_0+x_1+x_2\nn\\
          & = -2x_2+1,
    \end{align}
    where the last line provides the other inequality we were looking for. Hence $(\mathsf{h_3'})^\top\mathsf{x''}\leq0$ resulting in ${\mathrm{d^2}h_3\over\mathrm{d}z^2}\leq0$.
\end{proof}

\bibliographystyle{unsrt}


\end{document}